\documentclass[11pt]{article}
\usepackage[a4paper, margin=1in]{geometry}

\usepackage{amsmath,amsthm}
\usepackage{hyperref, url,afterpage}
\usepackage{float}

\usepackage{booktabs, array}
\usepackage{enumerate}
\usepackage[ruled]{algorithm}
\usepackage[commentColor=black]{algpseudocodex}
\usepackage{enumitem}
\usepackage[margin=1cm]{caption}

\newcommand{\hide}[1]{}

\newcommand{\mComment}[1]{\LComment{#1}}
\newcommand{\mlComment}[1]{\Comment{#1}}
\newcommand{\lineComment}[1]{\LComment{ #1}}
\newcommand{\rightComment}[1]{\Comment{ #1}}

\begin{document}

\newtheorem{theorem}{Theorem}
\newtheorem{observation}[theorem]{Observation}
\newtheorem{corollary}[theorem]{Corollary}
\newtheorem{lemma}[theorem]{Lemma}
\newtheorem*{lemma*}{Lemma}
\newcounter{factc}
\newtheorem{fact}[factc]{Fact}
\newtheorem*{theorem*}{Theorem}
\newtheorem{definition}{Definition}

\theoremstyle{remark}
\newtheorem{note}{Note}
\newtheorem*{problem}{Problem}

\title{Improved Approximation Bounds for Minimum Weight Cycle in the CONGEST Model\footnote{Authors' affiliation: Department of Computer Science, The University of Texas at Austin, Austin, TX, USA; email:  {\tt vigneshm@cs.utexas.edu, vlr@cs.utexas.edu}. This work was supported in part by NSF grant CCF-2008241.}
\author {Vignesh Manoharan  and Vijaya Ramachandran
}
}

\maketitle

\begin{abstract}

    Minimum Weight Cycle (MWC) is the problem of finding a simple cycle of minimum weight in a graph $G=(V,E)$. This is a fundamental graph problem with classical sequential algorithms that run in $\tilde{O}(n^3)$ and $\tilde{O}(mn)$ time\footnote[2]{\label{fn:tilde}We use $\tilde{O}$, $\tilde{\Omega}$ and $\tilde{\Theta}$  to absorb $polylog (n)$ factors.} where 
 $n=|V|$ and $m=|E|$. In recent years this problem has received significant attention in the context of fine-grained sequential complexity~\cite{williams2010subcubic,agarwal2018finegrained} as well as in the design of faster sequential approximation algorithms~\cite{kadria2023improved,kadria2022algorithmic,chechik2021mwc,harbuzova2024girth}, though not much is known in the distributed CONGEST model.
      
     We present sublinear-round approximation algorithms for computing MWC in directed graphs, and weighted graphs. Our algorithms use a variety of techniques in non-trivial ways, such as in our approximate directed unweighted MWC algorithm that efficiently computes BFS from all vertices restricted to certain implicitly computed neighborhoods in sublinear rounds, and in our weighted approximation algorithms that use unweighted MWC algorithms on scaled graphs combined with a fast and streamlined method for computing multiple source approximate SSSP. We present $\tilde{\Omega}(\sqrt{n})$ lower bounds for arbitrary constant factor approximation of MWC in directed graphs and undirected weighted graphs. 
 \end{abstract}
\newpage
\section{Introduction}

We present algorithms and lower bounds to compute a minimum weight cycle in the distributed CONGEST model.
Given a graph $G=(V,E)$ with a non-negative weight $w(e)$ on each edge $e\in E$, the minimum weight cycle problem (MWC)  asks for a cycle of minimum weight in $G$. An $\alpha$-approximation algorithm $(\alpha > 1)$ for MWC must find a cycle whose weight is within an $\alpha$ multiplicative factor of the true MWC. In the distributed setting, cycles are an important feature in network analysis with connections to deadlock detection and computing a cycle basis~\cite{peleg2013girth,Fraigniaud19distr,oliva18distr},
and a shortest cycle can model the likelihood of deadlocks in routing or in database applications~\cite{levitin10distr}. 
 
In the sequential context, MWC is a fundamental and well-studied problem on both directed and undirected graphs, both weighted and unweighted. MWC has classical sequential algorithms running in $\tilde{O}(n^3)$ time and $\tilde{O}(mn)$ time\footnotemark[2], where $|V|=n$ and $|E|=m$. There are also sequential fine-grained hardness results: MWC is in the $n^3$ class~\cite{williams2010subcubic}  and in the  $mn$ class~\cite{agarwal2018finegrained}  for hardness in graph path problems. In the distributed setting there are near-optimal results in the CONGEST model for most of the graph path problems in the sequential $n^3$ and $mn$ classes including APSP~\cite{bernsteinapsp,nanongkai2014approx}, radius and eccentricities~\cite{abboud2016radiusupper,abboud2016lower}, betweenness centrality~\cite{hoang2019round}, replacement paths and second simple shortest path~\cite{rp2024}, but very little was known for MWC prior to our work.

In {directed graphs}, exact MWC in the CONGEST model can be computed in $\tilde{O}(n)$ rounds by computing APSP~\cite{lenzen2019distributed,bernsteinapsp} and computing the minimum among cycles formed by concatenating a $v$-$u$ shortest path and a single edge $(u,v)$. 
This is matched by a nearly optimal $\tilde{\Omega}(n)$ lower bound for exact computation of MWC that we present in~\cite{rparxiv}, and the lower bound even applies to any $(2-\epsilon)$-approximation algorithm for MWC for any constant $\epsilon > 0$.
For an arbitrarily large $\alpha$-approximation (constant $\alpha \ge 2$), we show an $\tilde{\Omega}(\sqrt{n})$ lower bound. We complement this lower bound with sublinear approximation algorithms, with a non-trivial $\tilde{O}(n^{4/5}+D)$-round algorithm for computing a 2-approximation of directed unweighted MWC, and a $(2+\epsilon)$-approximation of directed weighted MWC.

In {undirected unweighted graphs}, where MWC is also known as \textit{girth}, the current best upper and lower bounds for exact computation in the CONGEST model are $O(n)$~\cite{holzer2012apsp} and $\tilde{\Omega}(\sqrt{n})$~\cite{frischknecht2012} respectively. For 2-approximation we have an algorithm announced in a previous paper~\cite{rparxiv} taking $\tilde{O}(\sqrt{n}+D)$ rounds, improving on the previous best upper bound $\tilde{O}(\sqrt{ng}+D)$~\cite{peleg2013girth} ($g$ is the weight of MWC). We show a lower bound of $\tilde{\Omega}(n^{1/3})$ for $(2.5-\epsilon)$-approximation and $\tilde{\Omega}(n^{1/4})$ for arbitrarily large constant $\alpha$-approximation.

For {undirected weighted graphs}, exact MWC can be computed in the CONGEST model using a reduction to APSP in $\tilde{O}(n)$ rounds~\cite{williams2010subcubic,agarwal2018finegrained} and we present a near-linear lower bound for $(2-\epsilon)$-approximation in~\cite{rparxiv}. In this paper, we prove a $\tilde{\Omega}(\sqrt{n})$ lower bound for $\alpha$-approximation (for any constant $\alpha \ge 2$). We complement this lower bound with an $\tilde{O}(n^{2/3}+D)$-round algorithm for $(2+\epsilon)$-approximation of MWC. 

Our approximation algorithms use a procedure to compute directed BFS or approximate SSSP from multiple sources, for which we provide a streamlined algorithm that is significantly more efficient than repeating the current best approximate SSSP algorithm.

An extended abstract of the results presented here appears in~\cite{mwc2024}.

\begin{table*}
    \caption{MWC results for CONGEST. Approximation results hold for approximation ratio $\alpha$ or $(1+\epsilon)$, where $\alpha > 1$ is an arbitrarily large constant, and $\epsilon >0$ is an arbitrarily small constant.}
    \label{tab:mwc}
    \centering
    \begin{tabular}{|c||c|c||c|c|}\hline
        \textbf{Problem} & \textbf{Lower Bound} & {\bf Ref.} & \textbf{Upper Bound} & {\bf Ref.} \\ \hline \hline
        \textit{Directed MWC} & $(2-\epsilon), \Omega\left(\frac{n}{\log n}\right)$ & \cite{rparxiv} & $1, \tilde{O}(n)$ & \cite{bernsteinapsp} \\
        \textit{unweighted or } & & & $2, \tilde{O}(n^{4/5} + D)$ (unweighted) & Thm~\ref{thm:dirub}.\ref{thm:dirub:unwt}  \\ 
        \textit{ weighted}& $\alpha, \Omega\left(\frac{\sqrt{n}}{\log n}\right)$ & Thm~\ref{thm:dirlb}.\ref{thm:dirlb:alphalb} &$(2+\epsilon), \tilde{O}( n^{4/5} + D)$ (weighted) &  Thm~\ref{thm:dirub}.\ref{thm:dirub:wt} \\ \hline \hline
        \textit{Undirected} & $(2-\epsilon), \Omega\left(\frac{n}{\log n}\right)$& \cite{rparxiv} & $1, \tilde{O}(n)$ &\cite{bernsteinapsp} \\ 
        {\it  weighted  MWC} & & & $(2+\epsilon), \tilde{O}(n^{2/3} + D)$ & Thm~\ref{thm:undirwt}.\ref{thm:undirwt:approxub} \\
           & $\alpha, \Omega\left(\frac{\sqrt{n}}{\log n}\right)$ &Thm~\ref{thm:undirwt}.\ref{thm:undirwt:alphalb} & &   \\ \hline
        \textit{Undirected } & $(2-\epsilon), \Omega\left(\frac{\sqrt{n}}{\log n}\right)$& \cite{frischknecht2012}& $1, {O}(n)$ &\cite{holzer2012apsp} \\ 
        {\it unweighted MWC}     & $(2-O(1/g)), \Omega\left(\frac{\sqrt{(n/g)}}{\log n}\right)$ & \cite{frischknecht2012}  & $(2-\frac{1}{g}), \tilde{O}(\sqrt{ng}+D)$&\cite{peleg2013girth} \\ 
        {\it (Girth)  }      & $(2.5-\epsilon), \Omega\left(\frac{n^{1/3}}{\log n}\right)$ & Thm~\ref{thm:undirunwt}.\ref{thm:undirunwt:alphalb2} & $(2-\frac{1}{g}), \tilde{O}(\sqrt{n}+D)$& Thm~\ref{thm:undirunwt}.\ref{thm:undirunwt:ub}\\ 
             & $\alpha, \Omega\left(\frac{n^{1/4}}{\log n}\right)$ & Thm~\ref{thm:undirunwt}.\ref{thm:undirunwt:alphalb} & & \\ \hline
    \end{tabular}
\end{table*}

\subsection{Preliminaries}
\label{sec:prelim}

\noindent
\textbf{The CONGEST Model.}\label{congest-model}
In the CONGEST model~\cite{peleg2000distributed}, a communication network is represented by a graph $G=(V,E)$ where nodes model processors and edges model bounded-bandwidth communication links between processors. Each node has a unique identifier in  $\{0, 1, \dots n-1\}$ where $n = |V|$, and each node only knows the identifiers of itself and its neighbors in the network. Each node has infinite computational power. The nodes perform computation in synchronous rounds, where each node can send a message of up to $\Theta(\log n)$ bits to each neighbor and can receive the messages sent to it by its neighbors. The complexity of an algorithm is measured by the number of rounds until the algorithm terminates. 

The graph $G$ can be directed or undirected but the communication links are always bi-directional (undirected) and unweighted; this follows the convention for CONGEST algorithms~\cite{chechiksssp,forster2018sssp,agarwal2019deterministic,ancona2020}. We consider algorithms on both weighted and unweighted graphs $G$ in this paper, where in weighted graphs each edge has a weight assignment $w:E(G) \rightarrow \{0,1,\dots W\}$ where $W=poly(n)$, and the weight of an edge is known to the vertices incident to it. Our algorithms readily generalize to larger edge weights in networks with bandwidth $\Theta(\log n + \log W)$, with our round complexities for weighted graphs having an additional factor of $\log (nW)$ for arbitrary integer $W$. 
The undirected diameter of the network, which we denote by $D$, is an important parameter in the CONGEST model.

In our algorithms, we frequently use the well-known broadcast and convergecast CONGEST operations~\cite{peleg2000distributed}: Broadcasting $M$ messages in total to all nodes, where each message could originate from any node, can be performed in $O(M+D)$ rounds. In the convergecast operation, each node holds an $O(\log n)$-bit value and we want to compute an associative operation (such as minimum or maximum) over all values. This can be performed in $O(D)$ rounds, after which all nodes know the result of the operation.
We now define the minimum weight cycle problem considered in this paper.
\begin{definition}\label{def:mwc}
    \textbf{Minimum Weight Cycle problem (MWC)}: Given an $n$-node graph $G=(V,E)$ ($G$ may be directed or undirected, weighted or unweighted), compute the weight of a shortest simple cycle in $G$. 
    In the case of $\alpha$-approximation algorithms, we need to compute the weight of a cycle that is within a factor $\alpha$ of the minimum (for $\alpha > 1$).
\end{definition}

In our distributed CONGEST algorithms for MWC, at the end of execution, every node in the network knows the computed weight of MWC (or approximate weight in case of approximation algorithm). Our CONGEST lower bounds for MWC apply even when only one node is required to know the weight of MWC. Our algorithms also allow us to construct the cycle by storing the next vertex on the cycle at each vertex that is part of the MWC.

\subsection{Our Results}
Table~\ref{tab:mwc} summarizes our upper and lower bound results.
All of our lower bounds hold for randomized algorithms, and the algorithms we present are also randomized -- which are correct with high probability in $n$. 

\subsubsection{ Directed Graphs.}
We have linear lower bounds for computing MWC in directed weighted or unweighted graphs, that also apply to  $(2-\epsilon)$-approximation~\cite{rparxiv}.
We address the problem of larger approximations, with an $\tilde{\Omega}(\sqrt{n})$ lower bound for $\alpha$-approximation of directed MWC, for arbitrarily large constant $\alpha \ge 2$. Our major algorithmic result is a sublinear round algorithm for computing 2-approximation of MWC in directed unweighted graphs that runs in $\tilde{O}(n^{4/5} + D)$ rounds, which we extend to a $(2+\epsilon)$-approximation algorithm in directed weighted graphs with the same round complexity. These results show that a linear lower bound is not possible for $\alpha \ge 2$ approximations. 

\begin{theorem} \label{thm:dirub} \label{thm:dirlb} 
Let $G=(V,E)$ be a directed graph.
In the CONGEST model, for any constants $\epsilon > 0, \alpha\ge 2$:
    \begin{enumerate}[label=\Alph*.,ref=\Alph*]
        \item \label{thm:dirlb:alphalb} Computing an $\alpha$-approximation of directed MWC (weighted or unweighted) requires $\Omega (\frac{\sqrt n}{\log n})$ rounds, even on graphs with diameter $D = \Theta (\log n)$.
        \item \label{thm:dirub:unwt} We can compute 2-approximation of unweighted MWC in  $\tilde{O}(n^{4/5} + D)$ rounds.
        \item \label{thm:dirub:wt} We can compute $(2+\epsilon)$-approximation of weighted MWC in $\tilde{O}(n^{4/5} + D)$ rounds.
    \end{enumerate}
\end{theorem}

\subsubsection{Undirected Unweighted Graphs.}

Girth (undirected unweighted MWC) can be computed in $O(n)$ rounds~\cite{holzer2012apsp} and an algorithm for $(2-\frac{1}{g})$-approximation that takes $\tilde{O}(\sqrt{ng}+D)$ rounds is given in~\cite{peleg2013girth} (where $g$ is the girth). A lower bound of $\Omega (\frac{\sqrt n}{\log n})$ for $(2-\epsilon)$-approximation of girth is given in~\cite{frischknecht2012}, and this proof also implies an $\Omega (\frac{\sqrt{(n/g)}}{\log n})$ lower bound for $(2-O(1/g))$-approximation of girth.

We improve on the result in~\cite{peleg2013girth} by presenting a faster $\tilde{O}(\sqrt{n} + D)$-round algorithm to compute $(2-\frac{1}{g})$-approximation of girth. Our algorithm was first announced in~\cite{rparxiv}, but we provide a more detailed description and analysis here.
For larger approximation ratios, we show a lower bound of $\tilde{\Omega} (n^{1/4})$ for arbitrarily large constant approximation. We also prove an intermediate result of $\tilde{\Omega} (n^{1/3})$ lower bound for $(2.5-\epsilon)$-approximation.

\begin{theorem} \label{thm:undirunwt} Consider an undirected unweighted graph $G=(V,E)$. In the CONGEST model:
    \begin{enumerate}[label=\Alph*.,ref=\Alph*]
        \item\label{thm:undirunwt:alphalb2} Computing a $(2.5-\epsilon)$-approximation of girth requires $\Omega (\frac{n^{1/3}}{\log n})$ rounds, even on graphs with constant diameter $D$.
        \item \label{thm:undirunwt:alphalb} For any constant $\alpha \ge 2.5$, computing an $\alpha$-approximation of girth requires $\Omega (\frac{n^{1/4}}{\log n})$ rounds, even on graphs with diameter $D = \Theta (\log n)$.
        \item \label{thm:undirunwt:ub} We can compute a $(2-\frac{1}{g})$-approximation of girth in $\tilde{O}(\sqrt{n} + D)$ rounds, where $g$ is the girth.
    \end{enumerate}
\end{theorem}

\subsubsection{Undirected Weighted Graphs.}
We have linear lower bounds for computing MWC in directed weighted or unweighted graphs, that also apply to  $(2-\epsilon)$-approximation~\cite{rparxiv}. Building on our method for the unweighted case, we present an algorithm for $(2+\epsilon)$-approximation of MWC that runs in $\tilde{O}(n^{2/3} + D)$ rounds.

\begin{theorem} \label{thm:undirwt} 
let $G=(V,E)$ be an undirected weighted graph $G=(V,E)$.
In the CONGEST model, for any constants $\epsilon > 0, \alpha \ge 2$:
    \begin{enumerate}[label=\Alph*.,ref=\Alph*]
        \item \label{thm:undirwt:alphalb} Computing an $\alpha$-approximation of MWC requires $\Omega (\frac{\sqrt n}{\log n})$ rounds, even on graphs with diameter $D = \Theta (\log n)$.
        \item \label{thm:undirwt:approxub} We can compute a $(2+\epsilon)$-approximation of MWC in $\tilde{O}(n^{2/3} + D)$ rounds.
    \end{enumerate}
\end{theorem}

\subsubsection{\texorpdfstring{Approximate $k$-source SSSP}{Approximate k-source SSSP}}
\label{sec:intro:ksssp}

A key subroutine in our approximation algorithms for MWC computes shortest paths efficiently from $k$ sources.

\begin{definition}\label{def:kbfs}
    \textbf{$k$-source BFS, SSSP problem}: Given an $n$-node graph $G=(V,E)$ and a set of $k$ vertices $U \subseteq V$, compute at each $v \in V$ the shortest path distance $d(u,v)$ for each $u \in U$. The problem is \textbf{$k$-source BFS} in unweighted graphs and \textbf{$k$-source SSSP} in weighted graphs.
\end{definition}

Optimal algorithms to compute $k$-source BFS taking $O(k+D)$ rounds are known for undirected unweighted graphs~\cite{lenzen2019distributed,hoang2019round}. For undirected weighted graphs, an algorithm in~\cite{elkin2019hopset} computes $(1+\epsilon)$-approximate SSSP in $\tilde{O}(\sqrt{nk}+D)$ rounds. When the number of sources is large, $k \ge n^{1/3}$, we present a fast and streamlined algorithm for directed exact BFS that runs in $\tilde{O}(\sqrt{nk}+D)$ rounds. We extend the algorithm to weighted graphs, computing approximate SSSP in both directed and undirected weighted graphs in $\tilde{O}(\sqrt{nk}+D)$ rounds. Our algorithm is more efficient than adapting the result in~\cite{elkin2019hopset} to directed weighted graphs for $k \ge n^{1/3}$, and matches their complexity for undirected graphs. In our applications to MWC algorithms, we only use the streamlined algorithm for $k \ge n^{1/3}$. We present our results for the complete range of $1 \le k \le n$ in Section~\ref{sec:ksssp}. In the following theorem, $SSSP = \tilde{O}(\sqrt{n}+n^{2/5+o(1)}D^{2/5}+D)$ refers to round complexity of computing SSSP (from a single source)~\cite{cao2023sssp}.

\begin{theorem}
    \label{thm:ksssp}
    \begin{enumerate}[label=\Alph*.,ref=\Alph*]
        \item We can compute exact directed BFS from $k$ sources in directed unweighted graphs with round complexity: \label{thm:ksssp:exact}
        $$
        \begin{cases}
            \tilde{O}(\sqrt{nk} + D) & ; k \ge n^{1/3} \quad (1)\\
            \min\left( \tilde{O}\left(\frac{n}{k}+D\right), k \cdot SSSP \right) & ; k<n^{1/3} 
        \end{cases}
        $$
        \item We can compute $(1+\epsilon)$-approximate weighted SSSP from $k$ sources in directed weighted graphs for any constant $\epsilon>0$ with round complexity:  \label{thm:ksssp:approx}
        $$
        \begin{cases}
            \tilde{O}(\sqrt{nk} + D) & ; k \ge n^{1/3} \;\; (2)\\
            \tilde{O}(\sqrt{nk} + k^{2/5}n^{2/5+o(1)}D^{2/5} + D) & ; k<n^{1/3}
        \end{cases}
        $$
    \end{enumerate}
\end{theorem}

\subsection{Significance of our Results}\label{sec:mwc-sig} 

In the distributed setting, cycles are an important network feature, with applications to deadlock detection and cycle basis computation~\cite{peleg2013girth,Fraigniaud19distr,oliva18distr}. In the sequential context, MWC is a fundamental graph problem that is well-studied. The $\tilde{O}(n^3)$ and $\tilde{O}(mn)$ time sequential algorithms for MWC have stood the test of time. MWC is in the sequential $n^3$ time fine-grained complexity class~\cite{williams2010subcubic} and plays a central role as the starting point of hardness for the $mn$ time fine-grained complexity class~\cite{agarwal2018finegrained}. 

The $n^3$ and $mn$ time fine-grained complexity classes contain important graph problems, which have been studied in the CONGEST model and for which nearly optimal upper and lower bounds have been obtained: All Pairs Shortest Paths (APSP)~\cite{bernsteinapsp}, Radius and Eccentricities~\cite{abboud2016lower,ancona2020}, Betweenness Centrality~\cite{hoang2019round}, Replacement Paths (RP) and Second Simple Shortest Path (2-SiSP)~\cite{rp2024}. However, there is a conspicuous lack of results for MWC (except for girth~\cite{holzer2012apsp,frischknecht2012,peleg2013girth} and reductions to APSP for exact MWC algorithms). 

In this paper, we make significant progress on this problem with a variety of results. While linear lower bounds hold for $(2-\epsilon)$-approximation, we present sublinear algorithms for computing $2$ or $(2+\epsilon)$-approximation of MWC. Our algorithms use a variety of techniques in non-trivial ways, such as our directed unweighted MWC algorithm that computes BFS from all vertices restricted to certain implicitly computed neighborhoods in sublinear rounds, and our weighted algorithms that use unweighted MWC algorithms on scaled graphs combined with multiple source approximate SSSP. We also present lower bounds for larger approximation factors, with $\tilde{\Omega}(\sqrt{n})$ bounds for arbitrarily large constant factor ($\ge 2$).

\subsection{Techniques}

\paragraph*{Lower Bounds.}

We establish $\tilde{\Omega}(\sqrt{n})$ lower bounds for {\it any} constant factor approximation algorithms for MWC in directed graphs by adapting a general lower bound graph used for MST, SSSP and other graph problems~\cite{sarma2012distributed,elkin2006mst} (Theorem~\ref{thm:dirlb}.\ref{thm:dirlb:alphalb}). We also adapt these constructions to obtain a similar lower bound for undirected weighted graphs (Theorem~\ref{thm:undirwt}.\ref{thm:undirwt:alphalb}). For undirected unweighted graphs, a lower bound of $\tilde{\Omega}(\sqrt{n})$ is known for $(2-\epsilon)$-approximation~\cite{frischknecht2012}, and we obtain an $\tilde{\Omega}(n^{1/4})$ lower bound for {\it any} constant factor approximation (Theorem~\ref{thm:undirunwt}.\ref{thm:undirunwt:alphalb}). We use reductions from set disjointness inspired by~\cite{drucker2014} to obtain our intermediate lower bound of $\tilde{\Omega}(n^{1/3})$ for $(2.5-\epsilon)$-approximation of undirected unweighted MWC.

\paragraph*{Approximate MWC Upper Bounds.} Our upper bounds use a framework of computing long cycles of high hop length and short cycles separately. Computing long cycles typically involves random sampling followed by computing shortest paths through sampled vertices. The sampling probability is chosen such that there is a sampled vertex on any long cycle with high probability, and we compute minimum weight cycles passing through sampled vertices. 

Computing short cycles requires a variety of techniques for each of our algorithms, with our method for {directed unweighted} MWC being the most involved. In a directed unweighted graph, we define a specific neighborhood for each vertex $v$ which contains a minimum weight cycle through $v$ if the MWC does not pass through any sampled vertex, a method inspired by the sequential algorithm of~\cite{chechik2021mwc}. In order to explore these neighborhoods efficiently, we perform a BFS computation with random scheduling from each vertex that is hop-restricted and restricted to the neighborhood. Additionally, to address congestion, we separately handle bottleneck vertices that send or receive a large number of messages.

The idea of bottleneck vertices has been used in~\cite{huang2017apsp,agarwal2020apsp} to control congestion while communicating information. This technique is used in the context of computing exact weighted APSP, in a subroutine to send distance information from all vertices to a set of sinks (common to all vertices) through shortest path trees that have been implicitly computed. 
They identify bottleneck vertices through which too many messages are to be sent, and compute distances through the bottleneck vertex separately. 
These algorithms can afford to use a super-linear number of rounds as they involve expensive Bellman-Ford operations, and they identify bottleneck vertices one at a time.
On the other hand in our algorithm, each vertex needs to send messages to a different set of vertices, i.e., the neighborhood containing short cycles. Another key difference is that our restricted BFS is performed on the fly simultaneously with identifying bottleneck vertices in a distributed manner, as we do not know the shortest path trees beforehand.
Finally, distances through all bottleneck vertices are computed with a pipelined hop-restricted BFS to maintain our sublinear round bound.

For {undirected unweighted graphs}, we compute the $\sqrt{n}$ closest neighbors efficiently using pipelining and compute cycles contained within them. We prove that for any cycle that extends outside a $\sqrt{n}$-neighborhood, a 2-approximation of this cycle is computed with a BFS from sampled vertices.

For {weighted graphs}, both directed and undirected, we use hop-bounded versions of the unweighted algorithms to compute approximations of short cycles. We use a scaling technique from~\cite{nanongkai2014approx}, where we construct a series of graphs with scaled weights such that distance-bounded shortest paths in the original graphs are approximated by some hop-bounded shortest path in a scaled graph.

\subsection{Prior Work}

\paragraph{Sequential Minimum Weight Cycle}
The problem of computing a minimum weight cycle in a given graph has been extensively studied in the sequential setting, for directed and undirected, weighted and unweighted graphs. It can be solved by computing All Pairs Shortest Paths (APSP) in the given graph in $O(n^3)$ time and in $\tilde{O}(mn)$ time. The hardness of computing MWC in the fine-grained setting was shown by~\cite{williams2010subcubic} for the $O(n^3)$-class and by \cite{agarwal2018finegrained} for the $O(mn)$ class. Fast approximation algorithms for computing MWC have been studied: \cite{chechik2021mwc} gives an $\tilde{O}(n^2)$ time and an $\tilde{O}(m\sqrt{n})$ time algorithm for computing a 2-approximation of directed MWC. For undirected unweighted graphs, \cite{itai1978mwc} computes MWC up to additive error 1 in $O(n^2)$ time, and a more general multiplicative $\alpha$-approximation can be computed in $\tilde{O}(n^{1+1/\alpha})$ time~\cite{kadria2022algorithmic}. For undirected weighted graphs, \cite{roditty2013mwc} computes a $\frac{4}{3}$-approximation in $\tilde{O}(n^2)$ time and \cite{kadria2023improved} computes a general $\frac{4}{3}\alpha$-approximation in $\tilde{O}(n^{1+1/\alpha})$ time.

\paragraph{Distrbuted Minimum Weight Cycle}
Lower and upper bounds for exact computation of MWC in the CONGEST model were given in~\cite{rparxiv}, with near-linear in $n$ round complexity bounds for directed graphs and undirected weighted graphs. These results for exact MWC are nearly optimal up to polylog factors. The lower bounds also apply to a $(2-\epsilon)$-approximation of MWC, but no bounds were known for coarser approximation.

\paragraph{Distributed Girth}
Minimum Weight Cycle in undirected unweighted graphs or girth has been studied in the distributed CONGEST model. An $O(n)$ algorithm for computing girth was given in~\cite{holzer2012apsp}, and a $\tilde{\Omega}(\sqrt{n})$ lower bound for computing girth was given in~\cite{frischknecht2012} which even applies to any $(2-\epsilon)$-approximation algorithm. An approximation algorithm that nearly matches this lower bound was given in~\cite{rparxiv}, improving on a result of~\cite{peleg2013girth}, with a $\tilde{O}(\sqrt{n}+D)$-round algorithm for a $(2-\frac{1}{g})$-approximation (where $g$ is the girth). For exact computation of girth, the gap between lower and upper bounds has been a longstanding open problem. 

\paragraph{Fixed-Length Cycle Detection} A problem closely related to girth is undirected $q$-cycle detection, where we want to check if a graph has a cycle of a certain length $q$. Lower bounds for $q$-cycle detection with $\Omega(n)$ lower bound for odd $q \ge 5$ and sublinear lower bounds for even $q \ge 4$ are given in~\cite{drucker2014}. The case of $q=3$ or triangle detection has been studied extensively~\cite{izumi2017triangle,chang2021near}. There is a $\tilde{O}(n^{1/3})$ algorithm for  triangle enumeration~\cite{chang2021near} and this result also applies to directed graphs~\cite{Pettie2022}. For even $q \ge 4$, sublinear round algorithms are given in~\cite{eden2021sublinear}. For directed $q$-cycle detection ($q \ge 4$), a tight linear lower bound was given in~\cite{rparxiv}.

\paragraph{CONGEST results for APSP and related problems}
The CONGEST round complexity of APSP~\cite{nanongkai2014approx} has been extensively studied, with nearly optimal upper and lower bounds of $\tilde{O}(n)$~\cite{bernsteinapsp} and $\Omega(\frac{n}{\log n})$~\cite{nanongkai2014approx} respectively. Upper and lower bounds for some related problems that have sequential $O(n^3)$ and $O(mn)$ algorithms have been studied in the CONGEST model, such as diameter~\cite{abboud2016,ancona2020}, replacement paths and second simple shortest paths~\cite{rparxiv}, radius and eccentricities~\cite{abboud2016,ancona2020}, and betweenness centrality~\cite{hoang2019round}.

\paragraph{CONGEST results for SSSP}
Our algorithms use distributed SSSP as a basic building block, and the CONGEST round complexity of both exact and approximate SSSP has been extensively researched~\cite{elkin2006mst,nanongkai2014approx,forster2018sssp,chechiksssp,cao2021approximatesssp,cao2023sssp}. For exact or $(1+\epsilon)$-approximate SSSP, the best known upper and lower bounds are $\tilde{O}(n^{2/5+o(1)}D^{2/5} + \sqrt{n} + D)$~\cite{cao2023sssp} and $\Omega(\sqrt{n}+D)$~\cite{elkin2006mst,sarma2012distributed} respectively. 

\section{Lower Bounds}
\label{sec:lb}

Our lower bounds use reductions from set disjointness which has an unconditional communication lower bound. Set Disjointness is a two-party communication problem, where two players Alice and Bob are given $k$-bit strings $S_a$ and $S_b$ respectively. Alice and Bob need to communicate and decide if the sets represented by $S_a$ and $S_b$ are disjoint, i.e., whether there is no bit position $i$, $1 \le i \le k$, with $S_a[i] = 1$ and $S_b[i] = 1$. A classical result in communication complexity states that Alice and Bob must exchange $\Omega(k)$ bits even if they are allowed shared randomness~\cite{kushilevitzcomm, razborov1992, baryossef2004}. Lower bounds using such a reduction also hold against randomized algorithms.

In an earlier paper~\cite{rparxiv}, we established a $\tilde{\Omega}(n)$ lower bound for exact computation of MWC in directed weighted and unweighted graphs, and our results are also announced in~\cite{mwc2024}. In Section~\ref{sec:dirmwcconstlb} we establish $\tilde{\Omega}(\sqrt{n})$ lower bounds for {\it any} constant factor approximation algorithms for MWC in directed graphs by adapting a general lower bound graph used for MST, SSSP and other graph problems~\cite{sarma2012distributed,elkin2006mst}. In Section~\ref{app:lower}, we modify these constructions to also apply to undirected weighted graphs.

\subsection{Lower Bound for Constant Approximation of Directed MWC}\label{sec:dirmwcconstlb}

We now consider the round complexity of computing an $\alpha$-approximation for a larger constant $\alpha \ge 2$. We show a $\Omega(\sqrt{n}/\log n)$ lower bound for any constant $\alpha$-approximation of directed MWC, both weighted and unweighted. Our technique is to adapt a general lower bound construction for undirected graphs in~\cite{sarma2012distributed} to directed MWC.

\begin{proof}[Proof of Theorem~\ref{thm:dirlb}.\ref{thm:dirlb:alphalb}]
    We use the construction in Figure~\ref{fig:mwcapproxlb}. The graph $G$ is constructed as a balanced binary tree with leaves $u_0, u_1,\dots u_{\ell}$ of height $p=\log \ell$, i.e., $\ell = 2^p$. All edges in the tree are directed towards the leaves except for edges on the path from $u_\ell$ to the root. There are also $q$ paths of length $\ell$, with directed path $i$ consisting of vertices $\langle v_0^i,v_1^i,\dots v_\ell^i \rangle$. Let $S_a, S_b$ be an instance of Set Disjointness on $q$ bits. For each $1\le i\le q$, the edge $(u_0,v_0^i)$ is present if and only if bit $i$ of $S_a$ is 1. Similarly, the edge $(v_{\ell}^i,u_{\ell})$ is present if and only if bit $i$ of $S_b$ is 1. For the remaining path vertices, we always add edges $(v_j^i, u_j)$ for $0 < j < \ell, 1 \le i \le q$. Note that only the edges incident to $u_0$ are directed towards path vertices, and others are directed towards leaf vertices.

    In the underlying undirected network of $G$, when Alice (who knows $S_a$) controls vertex $u_0$ and Bob (who knows $S_b$) controls vertex $u_{\ell}$, Lemma 4.1 of \cite{sarma2012distributed} proves a lower bound of $\Omega\left(\min(2^p, \frac{q}{2p \log n})\right)$ for computing set disjointness. We will now prove that $G$ has a directed cycle if and only if $S_a,S_b$ are not disjoint, thus obtaining a lower bound for any $\alpha$-approximation of directed MWC.

    We have directed the edges such that any cycle in the graph must involve a path from the root to $u_0$ along with a path from $u_0$ to $u_\ell$. This is because the only directed outgoing paths from the root are towards leaves $u_j \ne u_{\ell}$, and all leaves except $u_0$ and $u_{\ell}$ are sink vertices with no outgoing edge. So, any directed cycle must go through $u_0$ and along one of the $v^i$ paths to reach $u_{\ell}$. Such a cycle exists if and only if both $S_a[i]$ and $S_b[i]$ are 1, which means the sets are not disjoint. 
    
    Now, we set $\ell = \sqrt{n}/{\log n}$ and hence $p=\log \ell = \frac{1}{2}\log n - \log \log n$, and we set $q = \ell \cdot 2p \log n$. This gives a lower bound of $\Omega\left(\min(2^p, \frac{q}{2p \log n})\right) = \Omega(\ell) = \Omega\left(\frac{\sqrt{n}}{\log n}\right)$ for any algorithm determining whether $G$ has a directed cycle. This lower bound applies to any $\alpha$-approximation of directed MWC, since it must distinguish between directed cycles of finite and infinite length. Note that $G$ has a total of $q\cdot \ell + 2\ell = \Theta(n)$ vertices and undirected diameter $2p+2$, since any two vertices can be connected through paths in the tree --- so the lower bound holds even for graphs of diameter $\Theta(\log n)$.
\end{proof}

\begin{figure}
    \centering
    \begin{minipage}{0.48\textwidth}
        \centering
        \tikzstyle{vertex}=[circle, draw=black,minimum size=20pt]
        \tikzstyle{vertexsm}=[circle, draw=black,minimum size=5pt]
        \scalebox{0.45}{
        \begin{tikzpicture}
        
            \node[vertexsm] (u00) at (7,0) {};
            \node[vertexsm] (u0p2) at (3,-2) {};
            \node[vertexsm] (u0p1) at (1,-3.5) {};
            \node[vertexsm] (u1p1) at (5,-3.5) {};
            \node[vertexsm] (u1d) at (11,-3.5) {};
    
            \node (ul) at (4,-1.5) {\dots};
            \node (ur) at (9,-1.5) {\dots};
            \node (V) at (8,-8) {\dots};
            \node (V) at (8,-7) {\dots};
            \node (V) at (8,-9) {\dots};
            \node (V) at (8,-10) {\dots};
    
            \node[vertex] (up0) at (0,-5) {$u_0$};
            \node[vertex] (up1) at (2,-5) {$u_1$};
            \node[vertex] (up2) at (4,-5) {$u_2$};

            \node[vertex] (upd) at (12,-5) {$u_\ell$};
    
            \node[vertex] (v10) at (0,-7) {$v_{0}^1$};
            \node[vertex] (v11) at (2,-7) {$v_{1}^1$};
            \node[vertex] (v12) at (4,-7) {$v_{2}^1$};
            \node[vertex] (v1d1) at (10,-7) {};
            \node[vertex] (v1d) at (12,-7) {$v_{\ell}^1$};
    
            \node[vertex] (v20) at (0,-8) {$v_{0}^2$};
            \node[vertex] (v21) at (2,-8) {$v_{1}^2$};
            \node[vertex] (v22) at (4,-8) {$v_{2}^2$};
            \node[vertex] (v2d1) at (10,-8) {};
            \node[vertex] (v2d) at (12,-8) {$v_{\ell}^2$};
    
            \node[vertex] (vg0) at (0,-10) {$v_{0}^q$};
            \node[vertex] (vg1) at (2,-10) {$v_{1}^q$};
            \node[vertex] (vg2) at (4,-10) {$v_{2}^q$};
            \node[vertex] (vgd1) at (10,-10) {};
            \node[vertex] (vgd) at (12,-10) {$v_{\ell}^q$};
    
            \path[draw,thick,->] (u0p2) edge (u0p1);
            \path[draw,thick,->] (u0p2) edge (u1p1);
            \path[draw,thick,->] (u0p1) edge (up0);
            \path[draw,thick,->] (u0p1) edge (up1);
            \path[draw,thick,->] (u1p1) edge (up2);
    
            \path[draw,thick,->,dashed] (up0) edge [bend right] (v20);
            \path[draw,thick,->,dashed] (up0) edge [bend right] (vg0);
    
            \path[draw,thick,<-] (up1) edge [bend left] (v11);
            \path[draw,thick,<-] (up1) edge [bend left] (v21);
            \path[draw,thick,<-] (up1) edge [bend left] (vg1);
            \path[draw,thick,<-] (up2) edge [bend left] (v12);
            \path[draw,thick,<-] (up2) edge [bend left] (v22);
            \path[draw,thick,<-] (up2) edge [bend left] (vg2);
    
            \path[draw,thick,<-,dashed] (upd) edge [bend left] (v1d);
            \path[draw,thick,<-,dashed] (upd) edge [bend left] (v2d);
    
            \path[draw,thick,->] (v10) edge (v11);
            \path[draw,thick,->] (v11) edge (v12);
            \path[draw,thick,->] (v12) edge (6,-7);
            \path[draw,thick,->] (v1d1) edge (v1d);
    
            \path[draw,thick,->] (v20) edge (v21);
            \path[draw,thick,->] (v21) edge (v22);
            \path[draw,thick,->] (v22) edge (6,-8);
            \path[draw,thick,->] (v2d1) edge (v2d);
    
            \path[draw,thick,->] (vg0) edge (vg1);
            \path[draw,thick,->] (vg1) edge (vg2);
            \path[draw,thick,->] (vg2) edge (6,-10);
            \path[draw,thick,->] (vgd1) edge (vgd);
    
            \path[draw,thick,->] (u00) edge  (5.5,-1);
            \path[draw,thick,<-] (u00) edge  (8.5,-1);
            \path[draw,thick,<-] (u00) edge  (8.5,-1);
            \path[draw,thick,->] (upd) edge  (u1d);
            \path[draw,thick,->] (u1d) edge  (10,-5);
            \path[draw,thick,->] (u1d) edge  (10,-2);
        \end{tikzpicture}
        }
        \caption{Directed MWC Lower Bound for $\alpha$-approximation }
        \label{fig:mwcapproxlb}
    \end{minipage}
    \begin{minipage}{0.48\textwidth}
        \centering
        \tikzstyle{vertex}=[circle, draw=black,minimum size=20pt]
        \tikzstyle{vertexsm}=[circle, draw=black,minimum size=5pt]
        \scalebox{0.45}{
        \begin{tikzpicture}
        
            \node[vertexsm] (u00) at (7,0) {};
            \node[vertexsm] (u0p2) at (3,-2) {};
            \node[vertexsm] (u0p1) at (1,-3.5) {};
            \node[vertexsm] (u1p1) at (5,-3.5) {};
            \node[vertexsm] (u1d) at (11,-3.5) {};
    
            \node (ul) at (4,-1.5) {\dots};
            \node (ur) at (9,-1.5) {\dots};
            \node (V) at (8,-8) {\dots};
            \node (V) at (8,-7) {\dots};
            \node (V) at (8,-9) {\dots};
            \node (V) at (8,-10) {\dots};
    
            \node[vertex] (up0) at (0,-5) {$u_0$};
            \node[vertex] (up1) at (2,-5) {$u_1$};
            \node[vertex] (up2) at (4,-5) {$u_2$};

            \node[vertex] (upd) at (12,-5) {$u_L$};
    
            \node[vertex] (v10) at (0,-7) {$v_{0}^1$};
            \node[vertex] (v11) at (2,-7) {$v_{1}^1$};
            \node[vertex] (v12) at (4,-7) {$v_{2}^1$};
            \node[vertex] (v1d1) at (10,-7) {};
            \node[vertex] (v1d) at (12,-7) {$v_{L}^1$};
    
            \node[vertex] (v20) at (0,-8) {$v_{0}^2$};
            \node[vertex] (v21) at (2,-8) {$v_{1}^2$};
            \node[vertex] (v22) at (4,-8) {$v_{2}^2$};
            \node[vertex] (v2d1) at (10,-8) {};
            \node[vertex] (v2d) at (12,-8) {$v_{L}^2$};
    
            \node[vertex] (vg0) at (0,-10) {$v_{0}^g$};
            \node[vertex] (vg1) at (2,-10) {$v_{1}^g$};
            \node[vertex] (vg2) at (4,-10) {$v_{2}^g$};
            \node[vertex] (vgd1) at (10,-10) {};
            \node[vertex] (vgd) at (12,-10) {$v_{L}^g$};
    
            \path[draw,thick,-] (u0p2) edge (u0p1);
            \path[draw,thick,-] (u0p2) edge (u1p1);
            \path[draw,thick,-] (u0p1) edge (up0);
            \path[draw,thick,-] (u0p1) edge (up1);
            \path[draw,thick,-] (u1p1) edge (up2);
    
            \path[draw,thick,-,dashed] (up0) edge [bend right] node[left] {$1$} (v20);
            \path[draw,thick,-,dashed] (up0) edge [bend right] node[left] {$1$} (vg0);
    
            \path[draw,very thick,-] (up1) edge [bend left] node[left] {$\alpha n$} (v11);
            \path[draw,very thick,-] (up1) edge [bend left] node[right] {$\alpha n$} (v21);
            \path[draw,very thick,-] (up1) edge [bend left] node[right] {$\alpha n$} (vg1);
            \path[draw,very thick,-] (up2) edge [bend left] node[left] {$\alpha n$} (v12);
            \path[draw,very thick,-] (up2) edge [bend left] node[right] {$\alpha n$}(v22);
            \path[draw,very thick,-] (up2) edge [bend left] node[right] {$\alpha n$} (vg2);
    
            \path[draw,thick,-,dashed] (upd) edge [bend left] node[left] {$1$} (v1d);
            \path[draw,thick,-,dashed] (upd) edge [bend left] node[right] {$1$} (v2d);
    
            \path[draw,thick,-] (v10) edge (v11);
            \path[draw,thick,-] (v11) edge (v12);
            \path[draw,thick,-] (v12) edge (6,-7);
            \path[draw,thick,-] (v1d1) edge (v1d);
    
            \path[draw,thick,-] (v20) edge (v21);
            \path[draw,thick,-] (v21) edge (v22);
            \path[draw,thick,-] (v22) edge (6,-8);
            \path[draw,thick,-] (v2d1) edge (v2d);
    
            \path[draw,thick,-] (vg0) edge (vg1);
            \path[draw,thick,-] (vg1) edge (vg2);
            \path[draw,thick,-] (vg2) edge (6,-10);
            \path[draw,thick,-] (vgd1) edge (vgd);
    
            \path[draw,thick,-] (u00) edge  (5.5,-1);
            \path[draw,thick,-] (u00) edge  (8.5,-1);
            \path[draw,thick,-] (u00) edge  (8.5,-1);
            \path[draw,thick,-] (upd) edge  (u1d);
            \path[draw,thick,-] (u1d) edge  (10,-5);
            \path[draw,thick,-] (u1d) edge  (10,-2);
        \end{tikzpicture}
        }
        \caption{Undirected Weighted MWC Lower Bound for $\alpha$-approximation}
        \label{fig:mwcapproxunwlb}
    \end{minipage}
    \vspace{-0.1in}
    \caption*{Dotted lines refer to edges customized to set disjointness instance}
\end{figure}
\subsection{Undirected MWC Lower Bounds}
\label{app:lower}

\subsubsection{ Undirected Weighted MWC}
\label{sec:undirmwclb}

For exact computation of undirected weighted MWC, we established a $\tilde{\Omega}(n)$ lower bound in a previous paper~\cite{rparxiv}, and our results are announced in~\cite{mwc2024}. For $\alpha$-approximation, we adapt our directed MWC lower bounds in Section~\ref{sec:lb} to undirected weighted graphs to establish the $\tilde{\Omega}(\sqrt n)$ lower bound for any constant $\alpha$ (Theorem~\ref{thm:undirwt}.\ref{thm:undirwt:alphalb}).

\begin{proof}[Proof of Theorem~\ref{thm:undirwt}.\ref{thm:undirwt:alphalb}]
    We modify the directed MWC construction and make all edges undirected, shown in Figure~\ref{fig:mwcapproxunwlb}. We use weights to force the minimum weight cycle to use the edges corresponding to $S_a$ and $S_b$. Set all edge weights to one, except the edges $(u_1,v_1^i), (u_2,v_2^i), \cdots ,$ $(u_{L-1},v_{L-1}^i)$ for all $1\le i\le g$ which are set to weight $\alpha n$. This means that any cycle involving one of these high-weight edges has weight at least $\alpha n +1$. The only possible cycles using exclusively weight 1 edges use the path from root to $u_0$ and $u_L$ along with one of the paths $v^i$ -- such a cycle exists if and only if $S_a[i]$ and $S_b[i]$ are both 1. So, if the sets of not disjoint, there is a weighted cycle of weight $< n$. Otherwise, any cycle has weight at least $\alpha n+1$. This means any  $\alpha$-approximation algorithm for weighted undirected MWC can determine if $S_a,S_b$ are disjoint and we obtain an $\Omega\left(\frac{\sqrt{n}}{\log n}\right)$ round lower bound, using the approach in~\cite{sarma2012distributed}.
\end{proof}

\subsubsection{ Undirected Unweighted MWC}
In order to adapt our construction to undirected unweighted MWC, we replace the weighted edges in Figure~\ref{fig:mwcapproxunwlb} by unweighted paths. This increases the number of vertices in our lower bound construction, leading to a lower bound of $\tilde{\Omega}(n^{1/4})$ proving Theorem~~\ref{thm:undirunwt}.\ref{thm:undirunwt:alphalb}.

\begin{proof}[Proof of Theorem~\ref{thm:undirunwt}.\ref{thm:undirunwt:alphalb}]
    In Figure~\ref{fig:mwcapproxunwlb}, we replace each edge of weight $\alpha n$ by a path of length $\alpha n$, i.e., containing $\alpha n -1$ additional vertices. Thus, we obtain an undirected unweighted graph with $n' = n + \ell \cdot 2^p \cdot (\alpha n -1) = \Theta(\alpha n^2)$ vertices. Our reduction from set disjointness still holds as in the weighted case, giving a lower bound for $\Omega\left(\frac{\sqrt{n}}{\log n}\right)$. Since the new graph size if $n'$ with $n = \Theta(\sqrt{n})$ (since $\alpha$ is constant), we obtain a lower bound of $\Omega\left(\frac{n^{1/4}}{\log n}\right)$, proving our result.
\end{proof}

\begin{figure}
        \centering
        \tikzstyle{vertex}=[circle, draw=black,minimum size=20pt]
        \tikzstyle{vertexsm}=[circle, draw=black, fill=black, minimum size = 0.1cm, inner sep=0pt]
        \scalebox{0.8}{
            \begin{tikzpicture}
    
                \draw[blue] (1.5,-1.8) circle (2cm);
                \draw[blue] (6.5,-1.8) circle (2cm);

                \node (circle) at (0,0) {$G_k^a$};
                \node (circle) at (5,0) {$G_k^b$};

            
                \node[vertex] (ua) at (0.2,-1) {$u^a$};
                \node[vertex] (ub) at (5.2,-1) {$u^b$};

                \node[vertex] (va) at (1.2,-2.4) {$v^a$};
                \node[vertex] (vb) at (6.2,-2.4) {$v^b$};
            
            
        
            
                \node (topmid) at (4,1.5) {};
                \node (botmid) at (4,-5) {};
                
                \path[draw,thick,-] (ua) edge[bend left] (ub);
                \path[draw,thick,-] (va) edge[bend left] (vb);

                \path[draw,very thick,-] (ua) edge[dashed] node[pos=0.3,vertexsm] {}  node[pos=0.7,vertexsm] {} node[pos=0.5,right] {$t$} (va) ;
                \path[draw,very thick,-,dashed] (ub) edge  node[pos=0.3,vertexsm] {}  node[pos=0.7,vertexsm] {} node[pos=0.5,right] {$t$} (vb) ;
            
            
        
        
            
            
                \path[draw=blue,dashed] (topmid) edge (botmid);
            \end{tikzpicture}
        }
        \caption{Undirected Unweighted MWC Lower Bound for $2.5-\epsilon$-approximation}
        \label{fig:mwcapproxunwlb2}
    \vspace{-0.1in}
    \caption*{Dotted lines refer to edges customized to set disjointness instance}
\end{figure}

Additionally, we present a $\tilde{\Omega}(n^{1/3})$ lower bound for $(2.5-\epsilon)$-approximation,proving Theorem~\ref{thm:undirunwt}.\ref{thm:undirunwt:alphalb2} using a different technique inspired by constructions of~\cite{drucker2014}. Our construction is based on the Erdos conjecture: There exists a (undirected unweighted) graph $G_k=(V,E)$ such that $m = |E| = O(n^{1+1/k})$ and $G_k$ has girth $\ge 2k+1$. We have explicit constructions for $k=1,2,3,5$. 

We construct graph $G'$ made of two copies of $G_k$, denoted $G_k^a$ and $G_k^b$, and we denote $u^a \in V(G_k^a), u^b \in V(G_k^b)$ to be copies of the vertex $u \in V(G_k)$. We number the $m$ edges of $G_k$ from $1$ to $m$. Each edge in $G_k^a$, $G_k^b$ is replaced by a path of $t$ edges (constant $t$ to be fixed later) to simulate $t$ weight edges in an unweighted graph construction.
We add edges $(u^a,u^b)$ for each $u \in V(G_k)$ (weight 1). Note that $G' = (V',E')$ has $n' = m\cdot t + n = O(n^{1+1/k})$ vertices.

We use a reduction from $m$-bit set disjointness. Given set disjointness instance $S_a, S_b \in \{0,1\}^m$, for each $1 \le i \le m$, edge $i$ is present (as $t$ edge path) in $G_k^a$ iff $S_a[i]=1$ and edge $i$ is present (as $t$ edge path) in $G_k^b$ iff $S_b[i]=1$. Alice owns $G_k^a$ while Bob owns $G_k^b$. Let $g$ be the girth of $G'$. We prove the following lemma which is the basis of our reduction.

\begin{lemma}
    If $\exists i, S_a[i]=S_b[i]=1$, $g \le 2t + 2$. Otherwise, $g \ge (2k+1)t$
\end{lemma}
\begin{proof}
    First, assume $\exists i S_a[i]=S_b[i]=1$. Let edge $i$ have endpoints $(u,v)$ in $G_k$. The cycle $u^a$-$v^a$-$v^b$-$u^a$ has length $2t+2$.

    Now, assume otherwise that $\not\exists i S_a[i]=S_b[i]=1$. Let $C$ be a cycle of length $g$ in $G'$. If $C$ is completely contained within $G_k^a$ or $G_k^b$, then by construction of $G_k$, $C$ has at least $2k+1$ edges of $G_k$ and hence has length $(2k+1)t$ in $G'$.

    Assume $C$ contains vertices of both $G_k^a$ and $G_k^b$. $C$ can be decomposed into paths contained within $G_k^a$ and $G_k^b$ along with crossing edges between $G_k^a, G_k^b$. Say $C = v_1^a, \dots v_2^a, v_2^b \dots v_3^b,$ $v_3^a \dots v_4^a, \dots v_1^a$. Now consider mapping all paths contained within a copy to $G_k$, i.e. the path $C' = v_1,\dots v_2, \dots v_3, \dots v_4 \dots v_1$. Although all these edges may not be present in $G'$ since $S_a[i]$ or $S_b[i]$ is 0, all edges in $C'$ are valid edges in $G_k$ as either $S_a[i]$ or $S_b[i]$ was 1. This gives us a cycle of length $|C'|$ in $G_k$ hence $|C'| \ge 2k+1$. If cycle $C$ uses $s \ge 1$ crossing edges between $G_k^a, G_k^b$, then the length of $C$ in $G'$ is $|C'|t + s \ge (2k+1)t + 1$. Hence $g \ge (2k+1)t$.    
\end{proof}

Now, we complete the reduction. If there is an algorithm $\mathcal{A}$ for $(k+\frac{1}{2}-\epsilon)$-approximation of girth taking $R(n)$ rounds, Alice and Bob can simulate it in $O(R(n) \cdot n \log n)$ rounds as there are $2n$ crossing edges. If the sets $S_a$, $S_b$ are disjoint, the algorithm outputs value $\ge g = (2k+1)t$ as proven in the Lemma. If sets are not disjoint, the algorithm outputs value $\le (k+\frac{1}{2}-\epsilon)g = (k+\frac{1}{2}-\epsilon) \cdot (2t+2) = (2k+1)t + 2k+1-\epsilon\cdot (2t+2)$. We choose parameter $t$ such that $\epsilon\cdot (2t+2) > 2k+1$ ,i.e, $t = \Theta(1/\epsilon)$. Thus we can verify set disjointness in $O(R(n) \cdot n \log n)$ rounds by Alice and Bob simulating $\mathcal{A}$.

By set disjointness communication lower bound, $R(n) \cdot n\log n = \Omega(m) = \Omega(n^{1+1/k}) \Rightarrow R(n) = \Omega(\frac{n^{1/k}}{\log n})$. Note that $G'$ has $N = n^{(k+1)/k}$ vertices as we add a constant ($t$-1) number of edges to each of the $m$ edges ,i.e., $n = N^{k/(k+1)}$. So, in terms of graph size $N$, the lower bound is $R(N) = \Omega(\frac{N^{1/(k+1)}}{\log n})$.

Concretely, for $k=2$, we have lower bound of $\tilde{\Omega}(n^{1/3})$ for $(2.5-\epsilon)$-approximation. For $k=3$, we have lower bound of $\tilde{\Omega}(n^{1/4})$ for $(3.5-\epsilon)$-approximation but this is superseded by our earlier bound of  $\tilde{\Omega}(n^{1/4})$ lower bound for any $\alpha$-approximation. So, this technique does not provide non-trivial lower bounds for $k > 2$.

\section{\texorpdfstring{Multiple Source SSSP from $k \ge n^{1/3}$ sources}{Multiple Source SSSP from k>n1/3 sources}}
\label{sec:multiplesssp}

\begin{algorithm}[t]
    \caption{Exact $n^{1/3}$-source Directed BFS algorithm}
    \begin{algorithmic}[1]
        \Require Directed unweighted graph $G=(V,E)$, set of sources $U \subseteq V$ with $|U| = k = n^{1/3}$.
        \Ensure Every vertex $v$ computes $d(u,v)$ for each source $u \in U$.
        \State Let $h = n^{2/3}$. \label{alg:n3sssp:param} Construct set $S \subseteq V$ by sampling each vertex $v \in V$ with probability $\Theta(\frac{\log n}{h})$. W.h.p. in $n$, $|S| = \tilde{\Theta}(n^{1/3})$.\label{alg:n3sssp:samp}
        \State Compute $h$-hop directed BFS from each vertex in $S$. Repeat this computation in the reversed graph. This takes $O(|S| + h) = \tilde{O}(n^{2/3})$ rounds.\label{alg:n3sssp:sampbfs} \rightComment{Computes shortest path distances of $\le h$ hops.}
        \lineComment{The following lines compute ($>h$)-hop shortest path distances.}
        \State Construct a skeleton graph on vertex set $S$: For each directed $h$-hop shortest path in the underlying graph $G$ between sampled vertices found in line~\ref{alg:n3sssp:sampbfs}, add a directed edge with weight equal to shortest path distance.  \label{alg:n3sssp:skeleton} \rightComment{Internal computation}
        \State Share all edges of the skeleton graph by broadcast, node $v$ broadcasts all its outgoing edges. We broadcast up to $|S|^2$ values in total, which takes $O(|S|^2+D) = \tilde{O}(n^{2/3}+D)$ rounds. \label{alg:n3sssp:skeletonbroad}
        \State Each sampled vertex internally computes all pairs shortest paths in the skeleton graph using the broadcast values. \label{alg:n3sssp:skeletonapsp}
        \State Perform $h$-hop directed BFS from each source $u \in U$, in $O(h+k) = \tilde{O}(n^{2/3})$ rounds. If any sampled vertex $s \in S$ is visited during this BFS, $s$ broadcasts distance $d(u,s)$. We broadcast up to $k\cdot |S| = \tilde{\Theta}(n^{2/3})$ values, taking $\tilde{O}(n^{2/3}+D)$ rounds. \label{alg:n3sssp:kbfs}
        \State Using the broadcast information, sampled vertices determine their shortest path distance to sources in $U$: if distance $d(u,t)$ was broadcast for some $u \in U, t \in S$, each sampled vertex $s \in S$ locally sets $d(u,s) \gets \min(d(u,s), d(u,t)+d(t,s))$. \label{alg:n3sssp:sampsource}
        \State Each sampled vertex $s \in S$ propagates distance $d(u,s)$ for each $u \in U$ through $h$-hop BFS trees computed in line~\ref{alg:n3sssp:sampbfs}. Using random scheduling~\cite{ghaffarischeduling}, this takes $\tilde{O}(h+k|S|) = \tilde{O}(n^{2/3})$ rounds.\label{alg:n3sssp:propbfs}
        \State Each vertex $v$ receives distance $d(u,s)$ for source $u$ from a sampled vertex $s$ that contains $v$ in its $h$-hop BFS tree, and $v$ computes $d(u,v) \gets \min_{s \in S}( d(u,s) + d(s,v))$.
    \end{algorithmic}
    \label{alg:n3sssp}
\end{algorithm}

We present algorithms to compute directed $k$-source BFS in unweighted graphs and $k$-source SSSP in weighted graphs. We present an exact algorithm for directed unweighted BFS, which uses techniques of sampling and constructing a skeleton graph on sampled vertices, which are methods used in CONGEST single source reachability and SSSP algorithms~\cite{ghaffari2015reach,nanongkai2014approx}. For approximate $k$-source SSSP in undirected graphs, an algorithm taking $\tilde{O}(\sqrt{nk}+D)$ rounds was given in~\cite{elkin2019hopset}. We present an algorithm for $k$-source directed SSSP, utilizing a recent directed hopset construction from~\cite{cao2021approximatesssp}.

We start by presenting Algorithm~\ref{alg:n3sssp} for a directed unweighted graph $G=(V,E)$, computing $n^{1/3}$-source \textit{exact} directed BFS in $\tilde{O}(n^{2/3}+D)$ rounds, and later generalize our result to $k \ge n^{1/3}$ sources. We then extend our algorithm to weighted graphs to compute $k$-source $(1+\epsilon)$-approximate SSSP. 

Let $U \subseteq V$ be the set of sources. Algorithm~\ref{alg:n3sssp} first randomly samples a vertex set $S \subseteq V$ of size $\tilde{\Theta}(n^{1/3})$ in line~\ref{alg:n3sssp:samp}. We define a (virtual) skeleton graph on this vertex set $S$, where for vertices $u,v \in S$, an edge $(u,v)$ is added iff there is a directed path from $u$ to $v$ of at most $h=n^{2/3}$ hops in $G$. The skeleton graph is directed and weighted, with the weight of each skeleton graph edge being the $h$-hop bounded shortest path distance in $G$. The skeleton graph edges are determined using an $h$-hop directed BFS from each sampled vertex in line~\ref{alg:n3sssp:sampbfs}, and each sampled vertex uses this information to internally determine its outgoing edges in line~\ref{alg:n3sssp:skeleton}. These skeleton graph edges are then broadcast to all vertices in line~\ref{alg:n3sssp:skeletonbroad}. All pairs shortest path distances in the skeleton graph can be computed locally at each vertex in line~\ref{alg:n3sssp:skeletonapsp} using these broadcast distances, due to the chosen sampling probability.

In line~\ref{alg:n3sssp:kbfs}, an $h$-hop BFS is performed from each source, and each vertex that is at most $h$ hops from a source can compute its distance from that source. We now compute distances from each source to sampled vertices (regardless of hop-length) in line~\ref{alg:n3sssp:sampsource} using the $h$-hop bounded distances from line~\ref{alg:n3sssp:kbfs} along with skeleton graph distances from line~\ref{alg:n3sssp:sampbfs}. Finally distances from each source to all vertices are computed in line~\ref{alg:n3sssp:propbfs}, by propagating the distances computed in line~\ref{alg:n3sssp:sampsource} through the $h$-hop BFS trees rooted at each sampled vertex. This allows all vertices to locally compute their distance from all sources.

\begin{lemma}
    \label{lem:n3bfs}
    Algorithm~\ref{alg:n3sssp} computes exact directed BFS from $k=n^{1/3}$ sources in $\tilde{O}(n^{2/3} + D)$ rounds.
\end{lemma}
\begin{proof}
    \noindent 
    \textbf{\textit{Correctness:}} We show that our algorithm correctly computes distances so that each vertex $v \in V$ knows $d(u,v)$ for each source $u \in U$. In line~\ref{alg:n3sssp:sampbfs}, $h$-hop distances between sampled vertices are computed. Due to our sampling probability, any path of $h$-hops contains a sampled vertex w.h.p. in $n$. Thus, any shortest path between sampled vertices can be decomposed into $h$-hop bounded shortest paths that were computed in line~\ref{alg:n3sssp:sampbfs}. These $h$-hop bounded shortest paths make up the edges of the skeleton graph constructed in line~\ref{alg:n3sssp:skeleton}, and we broadcast all these distances in line~\ref{alg:n3sssp:skeletonbroad}. After line~\ref{alg:n3sssp:skeletonapsp}, all vertices have locally determined the shortest path distance between all pairs of sampled vertices.

    In line~\ref{alg:n3sssp:kbfs}, we compute $h$-hop directed BFS, $h=n^{2/3}$, from each of the $k=n^{1/3}$ sources, so all vertices within $h$ hops from a source have the correct distance. We now consider vertices further away --- consider a vertex $v$ such that a shortest path to source $u$ has more than $h$ hops. By our choice of sampling probability, there is a sampled vertex $s \in S$ on this shortest path at most $h$ hops from $v$ w.h.p. in $n$. If the shortest path from $u$ to $s$ has less than $h$ hops, then $s$ knows the distance $d(u,s)$ through the $h$-hop directed BFS in line~\ref{alg:n3sssp:kbfs} and propagates this distance to $v$ in line~\ref{alg:n3sssp:propbfs}. Otherwise, if the $u$-$s$ shortest path has more than $h$ hops, then w.h.p. in $n$ there is another sampled vertex $s' \in S$ that is on this shortest path at most $h$ hops from $u$. Thus, $s'$ knows the distance $d(u,s')$ in line~\ref{alg:n3sssp:kbfs} and broadcasts this distance. Now, vertex $s$ can compute distance $d(u,s) = \min_{s' \in S} d(u,s') + d(s',s)$. Thus, after line \ref{alg:n3sssp:sampsource} each $s\in S$ knows $d(u,s)$ for each $u \in U$. Now, since we had chosen $s$ such that $s$-$v$ shortest path has at most $h$ hops, the distance $d(u,s)$ is propagated to $v$ in line~\ref{alg:n3sssp:propbfs} and the distance $d(u,v) = \min_{s\in S} d(u,s) + d(s,v)$ is correctly computed.
    
    \noindent
    \textbf{\textit{Round Complexity:}} We perform BFS from $k$ sources restricted to $h$ hops in $O(k+h)$ rounds~\cite{lenzen2019distributed,hoang2019round} in lines~\ref{alg:n3sssp:sampbfs},\ref{alg:n3sssp:kbfs}. We also broadcast $M$ messages in $O(M+D)$ rounds~\cite{peleg2000distributed} in lines~\ref{alg:n3sssp:skeletonbroad},\ref{alg:n3sssp:kbfs}. For line~\ref{alg:n3sssp:propbfs} we used randomized scheduling~\cite{ghaffarischeduling} to pipeline the computation from all $|S|$ sampled vertices. We have total congestion $O(k|S|)$ (maximum number of messages through a single edge) since for each $s \in S$, propagating $k$ distances contributes $O(k)$ congestion. Our dilation is $O(h)$ as we propagate up to $h$ hops. So we can perform line~\ref{alg:n3sssp:propbfs} in $\tilde{O}(h+k|S|) = \tilde{O}(n^{2/3})$ rounds. The total round complexity of our algorithm is $\tilde{O}(n^{2/3} +D)$ rounds.
\end{proof}

We repeat the round complexity argument above with the following change in parameter to obtain our $k$-source directed BFS result.

\begin{lemma}
    We can compute exact directed BFS from $k\ge n^{1/3}$ sources in $\tilde{O}(\sqrt{nk} + D)$ rounds.
\end{lemma}
\begin{proof}
    We set parameter $h = \sqrt{nk}$ in line~\ref{alg:n3sssp:param} of Algorithm~\ref{alg:n3sssp}, and hence in line~\ref{alg:n3sssp:samp} we have $|S|=\tilde{\Theta}(\sqrt{\frac{n}{k}})$. From the round complexity analysis of Lemma~\ref{lem:n3bfs}, we get a bound of $\tilde{O}(h+k|S|+|S|^2+D)$, which is $\tilde{O}(\sqrt{nk}+\frac{n}{k}+D)$. Since we assume $k \ge n^{1/3}$, we have $\sqrt{nk} \ge \frac{n}{k}$ and hence the round complexity is $\tilde{O}(\sqrt{nk}+D)$.
\end{proof}

\paragraph*{\textbf{Weighted Graphs}}
\label{sec:n3ssspwt}
We extend our results to weighted graphs to compute $k$-source $(1+\epsilon)$-approximate SSSP. We use the following result from~\cite{nanongkai2014approx} that uses scaling to compute hop-bounded multiple source approximate SSSP. 

\begin{fact}[Theorem~3.6 of~\cite{nanongkai2014approx}]
    \label{fact:approxhopsssp}
    There is an algorithm that computes $(1 + \epsilon)$-approximate $h$-hop $k$-source shortest paths in a weighted (directed or undirected) graph $G=(V,E)$ in $O(\frac{\log n}{\epsilon} \cdot (k+h+D))$ rounds, where $D$ is the undirected diameter of $G$.
\end{fact}

\begin{lemma}
    \label{lem:n3ssspwt}
    We can compute $(1+\epsilon)$-approximate SSSP from $k\ge n^{1/3}$ sources in $\tilde{O}(\sqrt{nk} + D)$ rounds in directed and undirected weighted graphs.
\end{lemma}
\begin{proof}
    We replace the $\frac{n}{|S|}$-hop directed BFS computations in lines~\ref{alg:n3sssp:sampbfs},\ref{alg:n3sssp:kbfs},\ref{alg:n3sssp:propbfs} of Algorithm~\ref{alg:n3sssp} by approximate $\frac{n}{|S|}$-hop SSSP, using the algorithm stated in Fact~\ref{fact:approxhopsssp}. With this change, the distances computed in lines~\ref{alg:n3sssp:skeleton},\ref{alg:n3sssp:skeletonapsp},\ref{alg:n3sssp:sampsource},\ref{alg:n3sssp:propbfs} are $(1+\epsilon)$-approximations of the exact distances, and thus our final output is $(1+\epsilon)$-approximate shortest path distances. These changes to Algorithm~\ref{alg:n3sssp} only increase the round complexity by a factor $O(\frac{\log n}{\epsilon})$, so the round complexity of our $n^{1/3}$-source $(1+\epsilon)$-approximate SSSP is $\tilde{O}(\sqrt{nk} + D)$ in both directed and undirected weighted graphs.
\end{proof}

\paragraph*{\textbf{Multiple Source SSSP from $k < n^{1/3}$ sources}}
In unweighted graphs for $k < n^{1/3}$, Algorithm~\ref{alg:n3sssp} computes exact $k$-source BFS in $\tilde{O}(\frac{n}{k} + D)$ rounds by choosing parameter $h = \sqrt{nk}$. For small $k$, the simple algorithm of repeating SSSP computation in sequence from each source taking $k \cdot SSSP$ rounds is efficient (threshold for $k$ depends on value of $D$), giving the result in Theorem~\ref{thm:ksssp}.\ref{thm:ksssp:exact}. 

In Section~\ref{sec:ksssp}, we give an algorithm for approximate SSSP and BFS from $k < n^{1/3}$ sources with round complexity $\tilde{O}(\sqrt{nk} + k^{2/5}n^{2/5+o(1)}D^{2/5} + D)$, proving Theorem~\ref{thm:ksssp}.\ref{thm:ksssp:approx}. Our algorithm improves on the simple method of repeating the current best SSSP algorithm $k$ times for the entire range of $1 < k \le n$.

\section{Approximate Directed MWC}
\label{sec:dirmwcub}

We present a CONGEST algorithm for 2-approximation of directed unweighted MWC in Algorithm~\ref{alg:dirunwmwcub}. Our algorithm uses sampling combined with multiple source exact directed BFS (result (1) of Theorem~\ref{thm:ksssp}.\ref{thm:ksssp:exact}) to exactly compute the weight of MWC among long cycles of hop length $\ge h = n^{3/5}$. For the case when MWC is short with hop length $<h$, we use the subroutine detailed below in Section~\ref{sec:dirmwcshort} which efficiently performs a BFS from all vertices restricted to $h$ hops and to a specific size-bounded neighborhood for each vertex. 

In line \ref{alg:dirmwc:sample} of Algorithm~\ref{alg:dirunwmwcub}, we sample $\tilde{\Theta}(n^{2/5})$ vertices uniformly at random and in line~\ref{alg:dirmwc:ksssp}, we perform a directed BFS computation from each of them in $\tilde{O}(n^{7/10}+D)$ rounds (Theorem~\ref{thm:ksssp}.\ref{thm:ksssp:exact}). Using these computed distances, each sampled vertex locally computes a minimum weight cycle through itself in line~\ref{alg:dirmwc:longcyc}, thus computing MWC weight among long cycles. We use Algorithm~\ref{alg:dirunwmwcshort} (see Section~\ref{sec:dirmwcshort}) in line~\ref{alg:dirmwc:shortcyc} to handle short cycles. Computing short MWC requires the distances between all pairs of sampled vertices as input: so in line~\ref{alg:dirmwc:sampbroad} of Algorithm~\ref{alg:dirunwmwcub} we broadcast the $h$-hop shortest path distances between sampled vertices found during the BFS of line~\ref{alg:dirmwc:ksssp} and use these distances to locally compute shortest paths between all pairs of sampled vertices at each vertex. Thus, we exactly compute MWC weight if a minimum weight cycle passes through a sampled vertex in line~\ref{alg:dirmwc:longcyc}, and a 2-approximation of MWC weight in line~\ref{alg:dirmwc:shortcyc} otherwise. We now address the short cycle subroutine used in line~\ref{alg:dirmwc:shortcyc}.

\begin{algorithm}[t]
    \caption{$2$-Approximation Algorithm for Directed Unweighted MWC}
    \begin{algorithmic}[1]
        \Require Directed unweighted graph $G=(V,E)$
        \Ensure $\mu$, $2$-approximation of weight of a MWC in $G$ 
        \State Let $h=n^{3/5}$. Set $\mu_v \gets \infty$ for all $v \in V$. \rightComment{$\mu_v$ will track the minimum weight cycle through $v$ found so far.}
        \State Construct set $S$ by sampling each vertex $v \in G$ with probability $\Theta(\frac{1}{h} \cdot \log^3 n)$. W.h.p. in $n$, $|S| = \Theta( n^{2/5} \cdot \log^2 n )$. \label{alg:dirmwc:sample}
        \State Compute $d(s,v)$ for $s\in S,v\in V$ using multiple source exact directed BFS (Theorem~\ref{thm:ksssp}.\ref{thm:ksssp:exact}, Algorithm~\ref{alg:n3sssp}) from $S$. This takes $\tilde{O}(\sqrt{n|S|} + D)$ rounds as $|S|>n^{1/3}$. \label{alg:dirmwc:ksssp} 
        \State Compute cycles through $s \in S$:  For each edge $(v,s)$, $\mu_v \gets min(\mu_v,  w(v,s)+d(s,v))$. \label{alg:dirmwc:longcyc} \mlComment{Locally compute lengths of long cycles and all cycles passing through some sampled vertex.}
        \State Broadcast all pairs distances between sampled vertices: Each $t \in S$ broadcasts $d(s,t)$ for all $s \in S$. There are at most $|S|^2$ such distances, which takes ${O}(|S|^2 + D)$ rounds. \label{alg:dirmwc:sampbroad}
        \State Run Algorithm~\ref{alg:dirunwmwcshort} to compute approximate short MWC if it does not contain a sampled vertex, updating $\mu_v$ for each $v \in V$. \label{alg:dirmwc:shortcyc} 
        \rightComment{See Section~\ref{sec:dirmwcshort}.}
        \State Return $\mu \gets \min_{v \in V} \mu_v$, computed by a convergecast operation~\cite{peleg2000distributed} in $O(D)$ rounds. \label{alg:dirmwc:min}
    \end{algorithmic}
    \label{alg:dirunwmwcub}
\end{algorithm}

\subsection{Computing Approximate Short MWC}
\label{sec:dirmwcshort}

We present a method to compute 2-approximation of weight of minimum weight cycle among cycles of at most $h=n^{3/5}$ hops that do not pass through any sampled vertex in $S$. Our method is detailed in Algorithm~\ref{alg:dirunwmwcshort} and runs in $\tilde{O}(n^{4/5})$ rounds. As mentioned in the previous section, each vertex $v$ knows the distances $d(v,s), d(s,v)$ for each vertex $s \in S$, and distances $d(s,t)$ for $s,t \in S$.

\paragraph*{\textbf{Definition of $P(v)$ and $R(v)$}}
For each vertex $v \in V$, we define a neighborhood $P(v) \subseteq V$ such that $P(v)$ contains (w.h.p. in $n$) a cycle whose length is at most a 2-approximation of a minimum weight cycle $C$ through $v$ if $C$ does not pass through any sampled vertex. The construction of $P(v)$ is inspired by a sequential algorithm for directed MWC in~\cite{chechik2021mwc}, which uses the following lemma.

\begin{fact}[Lemma 5.1 of \cite{chechik2021mwc}]
    \label{fact:dirmwc}
    Let $C$ be a minimum weight cycle that goes through vertices $v$, $y$ in a directed weighted graph $G$. For any vertex $t$, if  $d(y,t) + 2d(v,y) \ge d(t,y) + 2d(v,t)$, then a minimum weight cycle containing $t$ and $v$ has weight at most $2w(C)$.
\end{fact}

Using Fact~\ref{fact:dirmwc}, we can eliminate vertices $y$ from $P(v)$ for which $d(y,t) + 2d(v,y) > d(t,y) + 2d(v,t)$ for some sampled vertex $t \in S$, as the minimum weight cycle through $v$ and $t$ is at most twice the minimum weight cycle through $v$ and $y$, and we compute MWC through all sampled vertices. 
We carefully choose a subset $R(v) \subseteq S$ of sampled vertices to eliminate vertices from $P(v)$ using Fact~\ref{fact:dirmwc}, such that $R(v)$ has size $\log n$ and the size of $P(v)$ is reduced to at most $\frac{n}{|S|} = \tilde{O}(n^{3/5})$. 

\begin{definition} \label{def:dirmwcpv}
    Given a subset $R(v) \subseteq S$, define 
    $$P(v) = \{ y \in V \mid \forall t \in R(v), d(y,t) + 2d(v,y) \le d(t,y) + 2d(v,t)\}$$
\end{definition}

Our algorithm for computing short MWC first partitions the sampled vertices into $\beta = \log n$ sets $S_1,\dots S_\beta$. In lines~\ref{alg:dirmwcshort:rstart}-\ref{alg:dirmwcshort:tvadd} of Algorithm~\ref{alg:dirunwmwcshort}, we construct $R(v)$ iteratively by adding at most one vertex from each $S_i$, so that $R(v)$ has size $\le \log n$. In the $i$'th iteration, we identify the vertices in $S_i$ that have not been eliminated from $P(v)$ by any of the $(i-1)$ previously chosen vertices in $R(v)$ and choose one of these vertices at random to add to $R(v)$. This computation is done locally at $v$ using the distances that have previously been shared to $v$.
We prove below in Lemma~\ref{lem:dirmwcshortub} that with this choice of $R(v)$, the size of $P(v)$ is at most $\frac{n}{|S|} = \tilde{\Theta}(n^{3/5})$, w.h.p. in $n$.

\paragraph*{\textbf{Restricted BFS}}
Our algorithm now computes directed BFS from each vertex $v$ restricted to $P(v)$. Let the MWC of $G$ have weight $\mu$, going through vertices $v$, $y$. If $y \in P(v)$, then the algorithm computes BFS from $v$ and computes the distance to $y$, thus outputting $\mu$ exactly. Otherwise, $y \not \in P(v)$ which means there exists a $t \in R(v)$ such that $d(y,t) + 2d(v,y) > d(t,y) + 2d(v,t)$. Then by Fact~\ref{fact:dirmwc}, the MWC through $t$ and $v$ has weight at most $2\mu$. Since the algorithm computes MWC through each sampled $t$, the output of the algorithm is at most $2\mu$. In order for our BFS to work, we need the following claim that each vertex $y \in P(v)$ can be reached by restricted BFS from $v$.

\begin{lemma}
    $P(v)$ induces a connected subgraph in the shortest path out-tree rooted at $v$. 
\end{lemma}
\begin{proof}
    Let vertex $y \in P(v)$. Then, we will prove that for any vertex $z$ on a shortest path from $v$ to $y$, $z \in P(v)$ thus proving our claim.

    We will use the fact that $d(v,y) = d(v,z) + d(z,y)$ by our assumption. Assume $y \in P(v)$, that is $\forall t \in R(v), d(y,t) + 2d(v,y) \le d(t,y) + 2d(v,t)$. Fix any $t \in R(v)$, we need to prove that $d(z,t) + 2d(v,z) \le d(t,z) + 2d(v,t)$. We will use the triangle inequalities $d(z,t) \le d(z,y)+ d(y,t)$ and $d(t,y) \le d(t,z) + d(z,y) \Rightarrow d(t,y) - d(z,y) \le d(t,z)$. 
    \begin{align*}
        d(z,t) + 2d(v,z) & \le d(z,y) + d(y,t) + 2d(v,z) \\
        & = d(y,t) + (2d(v,z) + 2d(z,y)) - d(z,y) \\
        & = d(y,t) + 2d(v,y) - d(z,y) \\
        & \le d(t,y) + 2d(v,t) - d(z,y) \; \left( \text{since } y \in P(v) \right)\\
        & = (d(t,y) - d(z,y)) + 2d(v,t) \\
        & \le d(t,z) + 2d(v,t)
    \end{align*}
\end{proof}

In lines~\ref{alg:dirmwcshort:bfsstart}-\ref{alg:dirmwcshort:bfsend} of Algorithm~\ref{alg:dirunwmwcshort}, we compute $h$-hop BFS from all vertices $v$ restricted to neighborhood $P(v)$: the BFS proceeds for $h$ steps, and at each step the BFS message is forwarded only to neighbors in $P(v)$. To test membership in $P(v)$ using Definition~\ref{def:dirmwcpv} at an intermediate vertex before propagating, we use distances between sampled vertices that is part of the input along with information about $R(v)$ that is included in the BFS message. 

The restricted BFS from every vertex needs to be carefully scheduled to obtain our sublinear round bound. We first implement random delays using ideas in~\cite{leighton1994packet,ghaffarischeduling}, where the start of BFS for a source $v$ is delayed by an offset $\delta_v$ chosen uniformly from range $[1,\rho = n^{4/5}]$ at random by $v$. Here, the parameter $\rho = n^{4/5}$ is chosen based on the maximum number of messages allowed throughout the BFS for a single vertex, as we shall see later. With this scheduling, all BFS messages for a particular source $v$ are synchronous even though messages from different sources may not be. We organize the BFS into phases, each phase involving at most $\Theta(\log n)$ messages.
However, the graph may contain bottleneck vertices $u$ that are in the neighborhood $P(v)$ for many $v$: such a vertex $u$ has to process up to $n$ messages, requiring $\Omega(n)$ phases even with random scheduling. 

We define $u$ to be a \textit{phase-overflow vertex} if it has to send or receive more than $\Theta(\log n)$ messages in a single phase of the restricted BFS. During the BFS, we identify any such phase-overflow vertex and terminate BFS computation through it. The BFS runs for $O(n^{4/5})$ phases and each BFS message contains $O(\log n)$ words, and hence the BFS takes a total of $\tilde{O}(n^{4/5})$ rounds. After this restricted BFS is completed, each vertex $u$ knows its shortest path distance from all vertices $v$ such that $u \in P(v)$ and there is a $v$-$u$ shortest path that has at most $h$ hops and contains no phase-overflow vertex.

Each phase takes $O(\log^2 n)$ CONGEST rounds, as a single BFS message contains $O(\log n)$ words (see line~\ref{alg:dirmwcshort:qv}). 
Thus, after the computation in lines~\ref{alg:dirmwcshort:bfsstart}-\ref{alg:dirmwcshort:bfsend}, each vertex $u$ knows its shortest path distance from all vertices $v$ such that $u \in P(v)$ and there is a $v$-$u$ shortest path that has $h$ hops and contains no phase-overflow vertex. It takes a total of  $\tilde{O}(n^{4/5})$ rounds to complete the $(h+\rho)=O(n^{4/5})$ phases in the computation.

Now, it remains to compute $h$-hop shortest path distances for paths that contain phase-overflow vertices.
We prove a bound of $\tilde{O}(n^{4/5})$ on the number of phase-overflow vertices as follows. We define $u$ to be a \textit{bottleneck vertex} if $u \in P(v)$ for more than $\rho=n^{4/5}$ vertices $v \in V$, i.e., $u$ may have to handle messages for more than $n^{4/5}$ sources across all phases of the BFS. We prove the following claims in Lemma~\ref{lem:dirmwcshortub}: (i) Vertex $u$ can be a phase-overflow vertex only if $u$ is a bottleneck vertex, and (ii) the number of bottleneck vertices is at most $\tilde{O}(n^{4/5})$. The bound on the number of bottleneck vertices is obtained using the bound of $\tilde{O}(n^{3/5})$ on the size of each $P(v)$. We compute $h$-hop directed BFS from the $\tilde{O}(n^{4/5})$ phase-overflow vertices in $O(h+n^{4/5})$ rounds.

After computing all distances from each $v \in V$ to vertices $y$ in $P(v)$, we locally compute the minimum among discovered cycles through $v$: at vertex $v$, a discovered cycle is formed by concatenating a $v$-$y$ shortest path and an incoming edge $(v,y)$.

\begin{algorithm}
    \caption{Approximate Short Cycle Subroutine}
    \begin{algorithmic}[1]
        \Require Directed unweighted graph $G=(V,E)$, set of sampled vertices $S \subseteq V$. Each vertex $v$ knows distances $d(v,s),d(s,v)$ for $s \in S$ and distances $d(s,t)$ for $s,t \in S$.
        \Ensure For each $v$, return $\mu_v$ which is a 2-approximation of minimum weight of cycles through $v$ among cycles that are short($<n^{3/5}$ hops) and do not pass through any vertex in $S$.  
        \State $h=n^{3/5}, \rho=n^{4/5}$. \label{alg:dirmwcshort:param}
        \State Partition $S$ into $\beta = \log n$ sets $S_1,\dots S_\beta$ of size $\Theta(n^{2/5} \cdot \log n)$. \label{alg:dirmwcshort:partition}
        \ForAll{vertex $v \in G$} \label{alg:dirmwcshort:rstart}
            \LComment{\textbf{Initial Setup}: Compute set $R(v)\subseteq S$, which is used to restrict BFS to neighborhood $P(v)\subseteq V$ (defined in Section~\ref{sec:dirmwcshort}).}
            \State $R(v) \gets \phi$
            \For{$i = 1 \dots \beta$} \mlComment{Local computation at $v$.}
                \State Let $T(v) = \{s \in S_i \mid \forall t \in R(v), d(s,t) + 2d(v,s) \le d(t,s)+2d(v,t)\}$. \label{alg:dirmwcshort:tv}
                \State If $T(v)$ is not empty, select a random vertex $s^* \in T(v)$ and add it to $R(v)$. \label{alg:dirmwcshort:tvadd}
            \EndFor
            \State $\delta_v$ is chosen uniformly at random from $\{1,\dots ,\rho\}$. \label{alg:dirmwcshort:delay} \mlComment{Choose BFS delay.}
            \State $Z(v) \gets 0$ \mlComment{$Z(v)$ is a flag that determines whether $v$ is a phase-overflow vertex.} \label{alg:dirmwcshort:zdef}
            \State Send $\{(d(v,s), d(s,v)) \mid s \in S\}$ to each neighbor $u$ in $O(|S|)$ rounds.  \label{alg:dirmwcshort:neighbbroad}
        \EndFor
        \LComment{\textbf{Restricted BFS from all vertices}: Computation is organized into phases where each vertex receives and sends at most $\log n$ BFS messages along its edges. Each BFS message contains $O(\log n)$ words and hence each phase takes $O(\log^2 n)$ CONGEST rounds.}
        \For{phase $r = 1 \dots (h+\rho)$} \label{alg:dirmwcshort:bfsstart}
            \ForAll{vertex $v \in G$}
                \If{$r = \delta_v$} \mlComment{This is the first phase for the BFS rooted at $v$.} \label{alg:dirmwcshort:initialbfs}
                    \State Construct message $Q(v) = (R(v), \{ d(v,t) \mid \forall t \in R(v)\})$ to be sent along the BFS rooted at $v$. $Q(v)$ contains $O(\log n)$ words ($|R(v)| \le \beta = \log n$) and can be sent in $O(\log n)$ rounds. \label{alg:dirmwcshort:qv}
                    \State Send BFS message $(Q(v), d(v,v)=0)$ to each out-neighbor of $v$. \label{alg:dirmwcshort:initialbfsend}
                \EndIf
                \mComment{Process and propagate messages from other sources $y$. We restrict the number of messages sent/received by a vertex by $\Theta(\log n)$ and identify phase-overflow vertices$(Z(v) \gets 1$) exceeding this congestion. Phase-overflow vertices are processed separately in line~\ref{alg:dirmwcshort:zsssp}.}
                \State Receive at most $\log n$ messages $(Q(y), d^*(y,v))$ from each in-neighbor. If more than $\Theta(\log n)$ messages are received from an edge, set $Z(v)\gets 1$ and terminate. \label{alg:dirmwcshort:zterminate}
                \State If message $(Q(y), d^*(y,v))$ is not the first message received for source $y$, discard it. Let $Y^r(v)$ denote the remaining set of sources $y$ with first time messages, and set $d(y,v) \gets d^*(y,v)$ for $y \in Y^r(v)$ . \label{alg:dirmwcshort:update}
                \State If $|Y^r(v)| > \Theta(\log n)$, set $Z(v) \gets 1$ and terminate. \label{alg:dirmwcshort:zterminateout}
                \State For each $y \in Y^r(v)$, and for each outgoing neighbor $u$, set estimate $d^*(y,u)\gets d(y,v)+1$. If $\forall t \in R(y), d(u,t) + 2d^*(y,u) \le d(t,u) + 2d(y,t)$, send message $(Q(y),d^*(y,u))$ to $u$. \rightComment{Note that $R(y), d(y,t)$ are known to $v$ from $Q(y)$ and $d(u,t),d(t,u)$ from line~\ref{alg:dirmwcshort:neighbbroad}.} \label{alg:dirmwcshort:transmit}
            \EndFor
        \EndFor \label{alg:dirmwcshort:bfsend}
    

        \mComment{Process phase-overflow vertices.}
        \State Let $Z = \{v \in V \mid Z(v) =  1\}$. Perform directed $h$-hop BFS with sources $Z$ in $O(|Z|+h)$ rounds. For each $v \in Z$ and edge $(x,v)$, set $\mu_x \gets min (\mu_x, d(v,x) + w(x,v))$. \label{alg:dirmwcshort:zsssp}
        \For{vertex $v \in V$}
            \State $\mu_v \gets \min(\mu_v, d(v,y)+1)$, for each $y \in V$ such that edge $(y,v)$ exists and $d(v,y)$ was computed during this algorithm. \label{alg:dirmwcshort:mincyc}
        \EndFor
    \end{algorithmic}
    \label{alg:dirunwmwcshort}
\end{algorithm}

\begin{lemma} \label{lem:dirmwcshortub}
    \begin{enumerate}[label=(\roman*)]
        \item Vertex $u \in V$ is a phase-overflow vertex only if $u$ is a bottleneck vertex. 
        \item There are at most $\tilde{O}(n^{4/5})$ bottleneck vertices w.h.p. in $n$.
        \item There are at most $\tilde{O}(n^{4/5})$ phase-overflow vertices w.h.p. in $n$. 
    \end{enumerate}
\end{lemma}
\begin{proof}
    We define $P^{-1}(u) = \{ v \in V \mid u \in P(v)\}$ to be the set of vertices for which $u$ is a part of their neighborhood. By definition, $u$ is a bottleneck vertex if $|P^{-1}(u)| \ge \rho = n^{4/5}$.

    \noindent
    \textit{Proof of (i)}: A message is sent to $u$ from source $v$ by some neighbor $x$ only if $u \in P(v)$, so $P^{-1}(u)$ is the set of sources for which $u$ has to send and receive BFS messages. Assume that $u$ is not a bottleneck vertex, $|P^{-1}(u)| < \rho$, we will prove that $u$ is not a phase-overflow vertex w.h.p. in $n$ , i.e., $u$ sends or receives at most $\Theta(\log n)$ messages in a single phase of BFS.
 
    By assumption, $u$ receives messages from at most $|P^{-1}(u)| < \rho$ sources throughout all phases of the restricted BFS. Additionally, each incoming edge to $u$ receives at most $\rho$ messages throughout the BFS since a single BFS sends at most one message through a single edge. Fix one such edge, that receives messages from sources $v_1,v_2, \dots v_{\gamma}$ for $\gamma < \rho$, and let the distance from $v_i$ to $u$ be $h_i$. The BFS messages from source $v$ are offset by a random delay $\delta_v \in \{1,2,\dots \rho\}$ and thus the BFS message from $v_i$ is received at $u$ at phase $h_i + \delta_{v_i}$. For a fixed phase $r$, the message from $v_i$ is sent to $u$ at phase $r$ iff $r = h_i +\delta_{v_i}$ which happens with probability $\frac{1}{\rho}$ since $\delta_{v_i}$ is chosen uniformly at random. Using a Chernoff bound, we can show that w.h.p. in $n$, there are at most $\Theta(\log n)$ of the $\gamma$ messages that are sent at phase $r$ through the chosen edge.
    
    Vertex $u$ sends an outgoing message for the BFS rooted at $v$ only if $u \in P(v)$ and it received a message from source $v$. So, $u$ sends at most $|P^{-1}(u)| < \rho$ outgoing messages through a single outgoing edge, and we can repeat the same argument above to argue that at most $\Theta(\log n)$ messages are sent out at a single phase. So, $u$ is not a phase-overflow vertex.

    \noindent
    \textit{Proof of (ii)}: We first argue that $P(v)$ has size at most $\frac{n}{|S|} = \tilde{\Theta}(n^{3/5})$ w.h.p in $n$ (adapting Lemma 6.2 of~\cite{chechik2021mwc}). When we add a vertex $t$ to $R(v)$, we expect $t$ to cover cycles through half the remaining uncovered vertices, since the condition we check (as in Definition~\ref{def:dirmwcpv}) is symmetric. At any iteration $i$, if the number of uncovered vertices is larger than $\tilde{\Theta}(n^{3/5})$, then with high probability there is some vertex in $A_i$ (which has size $\Theta(n^{2/5}\log n)$) that is also not covered and is added to $R(v)$, reducing the remaining number of uncovered vertices by half. So, the probability that the number of uncovered vertices $P(v)$ remains larger than $\tilde{\Theta}(n^{3/5})$ after $\log n$ such steps is polynomially small. 
    
    By definition of $P^{-1}(u)$, we have $\sum_{u\in V} |P^{-1}(u)| = \sum_{v \in V} |P(v)|$ (counting pairs of vertices $v, u \in P(v)$). Using the bound $|P(v)| \le \tilde{O}(n^{3/5})$, we get $\sum_{u\in V} |P^{-1}(u)| \le n \cdot \tilde{O}(n^{3/5})$. 
    
    Let $B$ denote the set of bottleneck vertices. Then, $\sum_{u\in V} |P^{-1}(u)| \ge \sum_{u \in B} |P^{-1}(u)| \ge |B| \cdot \rho$ and hence $|B| \le (n/\rho)\cdot \tilde{O}(n^{3/5}) = \tilde{O}(n^{4/5})$. 
    
    \noindent
    \textit{Proof of (iii)}: By \textit{(i)}, the number of phase-overflow vertices is $\le |B|$ and $|B| \le \tilde{O}(n^{4/5})$ by \textit{(ii)}.
\end{proof}

We now present details of the proof of correctness and round complexity of Algorithm~\ref{alg:dirunwmwcub}.

\begin{lemma}
    \label{lem:dirmwcub} Algorithm~\ref{alg:dirunwmwcub} correctly computes a 2-approximation of MWC weight in a given directed unweighted graph $G=(V,E)$ in $\tilde{O}(n^{4/5} + D)$ rounds.
\end{lemma}
\begin{proof}
    \noindent
    \textbf{{Correctness:}} Whenever we update $\mu_v$ for any $v \in V$, we use a shortest path from $x$ to $v$ along with an edge $(v,x)$, which means we only record weights of valid directed cycles. Let $C$ be a MWC of $G$ with weight $w(C)$ and let $v$ refer to an arbitrary vertex on $C$. In the following cases for $C$, Cases 1 and 2 are handled in Algorithm~\ref{alg:dirunwmwcub}, and Cases 3, 4 are handled by the subroutine in line~\ref{alg:dirmwc:shortcyc} using Algorithm~\ref{alg:dirunwmwcshort}. 

    \noindent
    {\boldmath\textbf{\textit{Case 1: $w(C) \ge h$:}}} $C$ contains at least $h$ vertices, and hence w.h.p. in $n$, $C$ contains a sampled vertex in $S$ by our choice of sampling probability. If $s \in S$ is on $C$, then the computation in line~\ref{alg:dirmwc:longcyc} exactly computes $w(C)$.

    \noindent
    {\boldmath\textbf{\textit{Case 2: $w(C) < h$ and $C$ extends outside $P(v)$:}}} Let $u$ be a vertex on $C$ such that $u \not\in P(v)$, then we have $d(u,t) + 2d(v,u) > d(t,u) + 2d(v,t)$ for some $t \in R(v)$. By Fact~\ref{fact:dirmwc}, this means that a minimum weight cycle containing $t$ and $v$ has weight at most $2w(C)$ since $C$ is a minimum weight cycle containing $v$ and $u$. Since $R(v) \subseteq S$, $t$ is a sampled vertex and hence $\mu_t \le 2w(C)$ by the computation in line~\ref{alg:dirmwc:longcyc}. Thus, we compute a 2-approximation of the weight of $C$.

    \noindent
    {\boldmath\textbf{\textit{Case 3: $w(C) < h$, $C$ is contained in $P(v)$, $\exists u \in C, Z(u) = 1 $:}}} In this case, $u$ is in the set of phase-overflow vertices $Z$ constructed in line~\ref{alg:dirmwcshort:zsssp} of Algorithm~\ref{alg:dirunwmwcshort}. After the BFS computation through vertices in $Z$, a minimum weight cycle through $u$ is computed in line~\ref{alg:dirmwcshort:zsssp}. Thus, $w(C)$ is computed exactly.

    \noindent
    {\boldmath\textbf{\textit{Case 4: $w(C) < h$, $C$ is contained in $P(v)$, $\forall u \in C, Z(u) = 0$:}}} In Algorithm~\ref{alg:dirunwmwcshort}, $u$ is not a phase-overflow vertex as $Z(u) = 0$ and $u$ never terminates its execution in line~\ref{alg:dirmwcshort:zterminate}. In fact, none of the vertices on $C$ terminate their execution and forward messages from all sources, including $v$. Thus, the vertex $z$ furthest from $v$ receives message $d(v,z)$ and records a cycle of weight $w(C)$ in line~\ref{alg:dirmwcshort:mincyc}.
    
    \noindent
    \textbf{{Round complexity:}} 
    We first address the running time of Algorithm~\ref{alg:dirunwmwcub} apart from line~\ref{alg:dirmwc:shortcyc} which invokes the subroutine Algorithm~\ref{alg:dirunwmwcshort}.

    We choose our sampling probability such that $|S| = \tilde{\Theta}(n/h) = \tilde{\Theta}(n^{2/5})$, so the multiple source SSSP in line~\ref{alg:dirmwc:ksssp} of Algorithm~\ref{alg:dirunwmwcub} takes time $\tilde{O}(\sqrt{n|S|}+D) = \tilde{O}(n^{7/10}+D)$ using Theorem~\ref{thm:ksssp}.\ref{thm:ksssp:exact} since we have $\tilde{\Theta}(n^{2/5}) > n^{1/3}$ sources. In line~\ref{alg:dirmwc:sampbroad}, we broadcast $|S|^2$ values taking $O(|S|^2+D) = \tilde{O}(n^{4/5}+D)$ rounds. Line~\ref{alg:dirmwc:min} involves a convergecast operation among all vertices, which takes $O(D)$ rounds~\cite{peleg2000distributed}.
    
    \noindent
    \textbf{\textit{Round complexity of Algorithm~\ref{alg:dirunwmwcshort}:}} We now show that Algorithm~\ref{alg:dirunwmwcshort} takes $\tilde{O}(n^{4/5})$ rounds. 
    
    The computation in lines~\ref{alg:dirmwcshort:param}-\ref{alg:dirmwcshort:zdef} is done locally at each vertex $v$. The local computation of $R(v)$ (lines~\ref{alg:dirmwcshort:rstart}-\ref{alg:dirmwcshort:tvadd}) only uses distances $d(v,t)$ and distances $d(s,t)$ for $s,t \in S$ that are part of the input. In line~\ref{alg:dirmwcshort:neighbbroad}, vertex $v$ sends $O(|S|)$ words of information to each neighbor, which takes $O(|S|)=\tilde{O}(n^{2/5})$ rounds. 
    
    We now address the round complexity of the restricted BFS of Lines~\ref{alg:dirmwcshort:bfsstart}-\ref{alg:dirmwcshort:bfsend}. The restricted BFS computation is organized into $(h+\rho)$ phases (recall $h=n^{3/5},\rho=n^{4/5}$). Each phase runs for $O(\log^2 n)$ rounds in which each vertex receives and sends up to $\Theta(\log n)$ BFS messages. Each message of the BFS is of the form $(Q(v),d(v,w))$ as in line~\ref{alg:dirmwcshort:qv}. Since $Q(v)$ has at most $\beta = \log n$ words, the BFS message can be sent across an edge in $O(\log n)$ rounds. The round bound for each phase is enforced in lines~\ref{alg:dirmwcshort:zterminate},\ref{alg:dirmwcshort:zterminateout} where propagation through a vertex is terminated if it has to send or receive more than $\Theta(\log n)$ messages in a single round, i.e., it is a phase-overflow vertex. The membership test in line~\ref{alg:dirmwcshort:transmit} is done using distances known to $v$ along with information from the BFS message, without additional communication. Thus, lines~\ref{alg:dirmwcshort:bfsstart}-\ref{alg:dirmwcshort:bfsend} take a total of $O\left((h+\rho) \cdot \log^2 n\right) = \tilde{O}(n^{4/5})$ rounds.
    
    We bound the round complexity of line~\ref{alg:dirmwcshort:zsssp} using Lemma~\ref{lem:dirmwcshortub} to bound the number of phase-overflow vertices by $\tilde{O}(n^{4/5})$, i.e., $|Z| \le \tilde{O}(n^{4/5})$. Now, the $h$-hop directed BFS in line~\ref{alg:dirmwcshort:zsssp} from $|Z|$ sources takes $O(|Z|+h) = \tilde{O}(n^{4/5})$ rounds~\cite{lenzen2019distributed}. Finally in line~\ref{alg:dirmwcshort:mincyc}, after all BFS computations are completed, we locally compute the minimum discovered cycle through each vertex $v$ formed by a $w$-$v$ shortest path along with edge $(v,w)$.
\end{proof}

\section{Undirected Unweighted MWC}
\label{sec:undirunwtmwcub}

In this section, we present an algorithm for computing $(2-\frac{1}{g})$-approximation of girth (undirected unweighted MWC) in $\tilde{O}(\sqrt{n}+D)$ rounds, where $g$ is the girth. We outline our method here, and present pseudocode in Appendix~\ref{app:undirmwc}. 

We first sample a set of $\tilde{O}(\sqrt{n})$ vertices and perform a BFS with each sampled vertex as source. For each non-tree edge $(x,y)$ in $T$, the BFS tree from sampled vertex $w$, we record a candidate cycle of weight $d(w,x)+d(w,y)+1$. Note that this may overestimate the size of the simple cycle $C$ containing edge $(x,y)$ and BFS paths by at most $2d(w,v)$ where $v$ is the closest vertex to $w$ on $C$ (i.e., $v$ is $lca(x,y)$ in $T$). For cycles where $d(w,v)$ is small relative to the size of $C$, we get a good approximation of the weight of the cycle. We prove that the only cycles for which a good approximation has not been computed are cycles entirely contained within the $\sqrt{n}$ neighborhood of each vertex in the cycle. We efficiently compute shortest path distances within each neighborhood using a source detection algorithm~\cite{peleg2000distributed}, and then compute MWC within the neighborhood.

If a minimum weight cycle extends outside the $\sqrt{n}$ neighborhood of even one of the vertices in the cycle, we show that a sampled vertex $w$ exists in this neighborhood. Thus, when we compute distances from each sampled vertex, we compute a 2-approximation of the weight of such a cycle. We use a more precise approach to obtain our $(2-\frac{1}{g})$-approximation, by computing lengths of cycles such that exactly one vertex is outside the neighborhood. The source detection procedure for $\sqrt{n}$-neighborhood takes $O(\sqrt{n}+D)$ rounds and BFS from $\tilde{O}(\sqrt{n})$ sampled vertices takes $\tilde{O}(\sqrt{n}+D)$ rounds giving us our total round complexity of $\tilde{O}(\sqrt{n}+D)$.

\paragraph*{\textbf{Computing $h$-hop limited MWC}} (used in Section~\ref{sec:undirwtmwcub}). If we are only required to compute approximate $h$-hop limited MWC, i.e. compute 2-approximation of minimum weight among cycles of $\le h$ hops, we can restrict our BFS computations to $h$ hops to obtain an $\tilde{O}(\sqrt{n}+h+D)$ round algorithm. This does not improve the running time for unweighted graphs as $h \le D$, but in Section~\ref{sec:undirwtmwcub} we will apply this procedure to weighted graphs using the following notion of stretched graph: given a network $G=(V,E)$ with weights on edges, a \textit{stretched} unweighted graph $G^s$ is obtained by mapping each edge of $G$ with weight $w$ to a unweighted path of $w$ edges. If $G$ is directed, the path is directed as well.

Given edge-weighted network $G=(V,E)$, we can efficiently simulate the corresponding stretched graph $G^s$ on the network by simulating all but the last edge of the path corresponding to a weighted edge at one of the endpoints. The diameter of the stretched graph may be much larger than that of $G$ but convergecast operations cost only $R_{cast} = O(D)$ rounds where $D$ is the undirected diameter of $G$. Thus, we can compute $h$-hop limited unweighted MWC in $G^s$ in $\tilde{O}(\sqrt{n}+h+R_{cast})$ rounds.
Note that a cycle of hop length $h$ in $G^s$ corresponds to a cycle of weight $h$ in $G$. See Appendix~\ref{app:undirmwc} for details. We use this idea in the next section with scaled-down weights so that even a cycle of large weight in $G$ can be approximated by a cycle of low hops in an appropriate stretched graph. 

\begin{corollary}
    \label{thm:hoplimapproxundirmwcub}
    Given a network $G=(V,E)$ with edge weights, we can compute a $(2-1/g)$-approximation of $h$-hop limited MWC of $G^s$ (stretched unweighted graph of $G$) in $\tilde{O}(\sqrt{n}+h+R_{cast})$ rounds, where $g$ is the $h$-hop limited MWC value and $R_{cast}$ is the round complexity of convergecast.
\end{corollary}

\section{Weighted MWC}
\label{sec:wtmwcub}
In this section, we present algorithms to compute $(2+\epsilon)$-approximations of weighted MWC in $\tilde{O}(n^{2/3}+D)$ for undirected graphs and $\tilde{O}(n^{4/5}+D)$ for directed graphs. Our algorithms use the $k$-source approximate SSSP algorithm from Section~\ref{sec:multiplesssp} along with unweighted MWC approximation algorithms of Sections~\ref{sec:dirmwcub} and~\ref{sec:undirunwtmwcub} on scaled graphs to approximate weighted MWC. 

\subsection{Approximate Undirected Weighted MWC}
\label{sec:undirwtmwcub}
We present a method that computes $(2+\epsilon)$-approximation in $\tilde{O}(n^{2/3} + D)$ rounds, proving Theorem~\ref{thm:undirwt}.\ref{thm:undirwt:approxub}.
We set parameter $h=n^{2/3}$. For long cycles with $\ge h$ hops, we sample $\tilde{\Theta}(n^{1/3})$ vertices so that w.h.p. in $n$ there is at least one sampled vertex on the cycle. We compute $(1+\epsilon)$-approximate SSSP from each sampled vertex using the multiple source SSSP algorithm (result (2) of Theorem~\ref{thm:ksssp}.\ref{thm:ksssp:approx}) to compute $(1+\epsilon)$-approximate MWC among long cycles. 

For small cycles (hop length $< h$), we use scaling along with a hop-limited version of 2-approximate undirected unweighted MWC algorithm from Corollary~\ref{thm:hoplimapproxundirmwcub} of Section~\ref{sec:undirunwtmwcub}. We use the scaling technique in~\cite{nanongkai2014approx}, where it was used in the context of computing approximate shortest paths. We construct $O(\log n)$ scaled versions of the graph $G$, denoted $G^i$ for $1 \le i \le \log(hW)$ with weight $w$ edge scaled to have weight $\lceil \frac{2hw}{\epsilon2^i} \rceil$. Each $h$-hop limited shortest path $P$ in $G$ is approximated by a path of weight at most $h^* = \left(\left(1+\frac{2}{\epsilon}\right) \cdot h\right)$ in some $G^{i^*}$ --- this $i^*$ is in fact $\lceil \log w(P) \rceil$ where $w(P)$ is the weight of $P$ in $G$, as proven in~\cite{nanongkai2014approx}.

We run an $h^*$-hop limited version of the unweighted approximate MWC algorithm on the stretched scaled graph (see Section~\ref{sec:undirunwtmwcub}). The stretched graph may have large edge weights and such edges may not always be traversed in $h^*$ rounds, but a $(1+\epsilon)$-approximation of any $h$-hop shortest path is traversed within $h^*$ rounds in at least one of the $G^i$. We apply Corollary~\ref{thm:hoplimapproxundirmwcub} to compute $2$-approximation of $h^*$-hop limited MWC in each stretched $G^i$, and take the minimum to compute $(2+\epsilon)$-approximate $h$-hop limited MWC in $G$.

\begin{fact}[Theorem~3.3 of~\cite{nanongkai2014approx}]
    \label{fact:nanongkaiscaling}
    Given a (directed or undirected) weighted graph $G=(V,E)$ and parameter $h$, construct a scaled weighted graph $G^i=(V,E)$, for each integer $1 \le i \le \log_{(1+\epsilon)} (h W$) ($W$ is maximum edge weight in $G$), with weight of edge $e \in E$ changed from $w(e)$ to $\left\lceil \frac{2h w(e)}{\epsilon 2^i} \right\rceil$.  Let $d_h(x,y)$ denote $h$-hop limited shortest path distances in $G$, and $d^i(x,y)$ denote distances in $G^i$. Let distance estimate $\tilde{d}_h(x,y)$ be defined as follows
    \begin{align*}
        \tilde{d}_h(x,y) = \min \left\{ \frac{\epsilon2^i}{2h} d^i(x,y) \middle\vert i : d^i(x,y) \le \left(1+\frac{2}{\epsilon}\right) h \right\}
    \end{align*}
    Then, $\tilde{d}_h(x,y)$ is a $(1+\epsilon)$-approximation of $d_h(x,y)$, i.e.,  $d_h(x,y) \le \tilde{d}_h(x,y) \le (1+\epsilon) d_h(x,y)$
\end{fact}

\begin{algorithm}[t]
    \caption{$(2+\epsilon)$-Approximation Algorithm for Undirected Weighted MWC}
    \begin{algorithmic}[1]
        \Require Undirected weighted graph $G=(V,E)$
        \Ensure $M$, $(2+\epsilon)$-approximation of MWC of $G$ 
        \State Let $h = n^{2/3}$, $h^* = \left(\left(1+\frac{2}{\epsilon}\right) \cdot h\right)$
        \State Construct set $S \subseteq V$ by sampling each vertex $v \in G$ with probability $\Theta(\frac{\log n}{h})$. W.h.p in $n$, $|S| = \tilde{\Theta}(n^{1/3})$.\label{alg:wtmwc:sample} 
        \State Perform $(1+\epsilon)$-approximate multiple source SSSP for each $w \in S$: each vertex $z \in V$ knows the distance $d(w,v)$. We also keep track of the vertex adjacent to source $w$ on each computed $w$-$v$ shortest path, denoted $f(w,v)$.\label{alg:wtmwc:long} 
        \State Each vertex $v$ shares its shortest path information $\{(d(w,v),f(w,v)) \mid w \in S\}$ with its neighbors.\label{alg:wtmwc:neighbsend}
        \ForAll{vertex $v \in V$} \mlComment{Local computation at $v$}
            \lineComment{Compute MWC among all long cycles}
            \State For each neighbor $x$ of $v$, for each sampled vertex $w \in S$, if $f(w,v) \ne f(w,x)$, then $v$ records a cycle through $w$,$v$ and $x$: $M_v \gets \min(M_v, d(w,v) + d(w,x) + w(v,x))$. \label{alg:wtmwc:longcyc} 
        \EndFor
        \For{$1 \le i \le \log_{(1+\epsilon)} hW$}\label{alg:wtmwc:scalestart}
            \lineComment{Construct scaled graphs and compute distance-limited MWC to compute short cycles}
            \State Replace every edge $e \in E$ of weight $w(e)$ with a path of length $w^{i}(e) = \left\lceil \frac{2h w(e)}{\epsilon 2^i} \right\rceil$ to get graph $G^i$.\label{alg:wtmwc:scale}
            \State Compute $h^*$-hop limited $2$-approximation of undirected unweighted MWC of $G^i$ using Algorithm~\ref{alg:undirunwmwcub} by simulating weighted edges as unweighted paths. Let $M^i$ denote this distance-limited MWC of $G^i$.\label{alg:wtmwc:scalemwc}
            \State $M \gets \min(M, \frac{\epsilon2^i}{2h} M^i)$. \label{alg:wtmwc:mwcrec} \rightComment{Reverse scaling to obtain hop-limited MWC in $G$}
        \EndFor 
        \State Return $M \gets \min(M, \min_{v \in V} M_v)$, by performing a convergecast operation~\cite{peleg2000distributed}. \label{alg:wtmwc:globalmin}
    \end{algorithmic}
    \label{alg:undirwtmwcub}
\end{algorithm}

\begin{lemma}
    Algorithm \ref{alg:undirwtmwcub} computes a $(2+\epsilon)$-approximation of MWC in $\tilde{O}(n^{2/3} + D)$ rounds.
    \label{lem:undirwtmwcub}
\end{lemma}
\begin{proof}
    \noindent
    \textbf{\textit{Correctness:}} Let $C$ be a minimum weight cycle of weight $w(C)$. We first address the case when $C$ is long, having more than $h=n^{2/3}$ hops.

    The sampled set $S$ has size $\tilde{O}(n/h)$ and w.h.p. in $n$, at least one sampled vertex is on any cycle of hop length $\ge h$. Let $w \in S$ be a sampled vertex on $C$. In line~\ref{alg:wtmwc:long}, the shortest path from $w$ to each $v \in V$ is computed. For any minimum weight cycle $C$ through vertex $w$, we can identify two vertices $v$ and $x$ such that $C$ is the concatenation of $w$-$x$ and $w$-$v$ shortest path along with edge $(w,v)$. This edge $(w,v)$ is called the critical edge of $C$~\cite{williams2010subcubic,agarwal2018finegrained}. The $w$-$x$ and $w$-$v$ shortest paths are such that the next vertex from $w$ on the two paths are distinct, as cycle $C$ is simple. This observation is used in line~\ref{alg:wtmwc:longcyc}, to correctly compute long cycles through each vertex using the first vertex information computed along with SSSP in line~\ref{alg:wtmwc:long} and the neighbor distance information sent in line~\ref{alg:wtmwc:neighbsend}.
    
    Lines~\ref{alg:wtmwc:scale}-\ref{alg:wtmwc:mwcrec} compute a $(2+2\epsilon)$-approximation of $w(C)$ if $C$ is short with hop length $< h$, and we can set $\epsilon' = 2\epsilon$ to get a $(2+\epsilon')$-approximation. To prove this, we use Fact~\ref{fact:nanongkaiscaling}, which implies the minimum of the $h^*$- hop limited unweighted MWC in $G^{i}$ over all $i$ is a $(1+\epsilon)$-approximation of the $h$-hop limited weighted MWC in $G$. Since the undirected unweighted MWC algorithm computes a $2$-approximation, we compute a $2(1+\epsilon)$-approximation of the weight of $C$ in line~\ref{alg:wtmwc:scalemwc}. Thus at some iteration $i$, $M$ is correctly updated with this weight in line~\ref{alg:wtmwc:mwcrec}. 
    
    Finally, the minimum weight cycle among the computed long and short cycles is computed by a global minimum operation in line~\ref{alg:wtmwc:globalmin}. 
    
    \noindent
    \textbf{\textit{Round Complexity:}}
    Computing the $(1+\epsilon)$-approximate SSSP from $\tilde{\Theta}(n/h) = \tilde{\Theta}(n^{1/3})$ sources in line~\ref{alg:wtmwc:long} takes $\tilde{O}(n^{2/3} + D)$ using Algorithm~\ref{alg:n3sssp}. The communication in line~\ref{alg:wtmwc:neighbsend} sends $O(|S|)$ words taking $\tilde{O}(n^{1/3})$ rounds.
    The loop in line \ref{alg:wtmwc:scalestart} runs for $O(\log n)$ iterations, and in each iteration, lines~\ref{alg:wtmwc:scale}-\ref{alg:wtmwc:mwcrec} take $\tilde{O}(\sqrt{n}+h+D)$ rounds by Corollary~\ref{thm:hoplimapproxundirmwcub} --- note that $R_{cast} = O(D)$ where $D$ is the undirected diameter of $G$ regardless of the diameter of scaled graph $G^i$.  The convergecast operation in line~\ref{alg:wtmwc:globalmin} takes $O(D)$ rounds~\cite{peleg2000distributed}. Thus, the total round complexity is $\tilde{O}(n^{2/3} + D)$.
\end{proof}

\subsection{Approximate Directed Weighted MWC}
\label{sec:dirwtmwcub}
We use the above framework for undirected graphs to compute $(2+\epsilon)$-approximation of directed weighted MWC, by replacing the hop-limited undirected unweighted MWC computation with a directed version. We can compute $h$-hop limited 2-approximation of MWC in stretched directed unweighted graphs in $\tilde{O}(n^{4/5}+h+R_{cast})$ rounds by applying the modifications in Corollary~\ref{thm:hoplimapproxundirmwcub} to our directed unweighted MWC algorithm in Section~\ref{sec:dirmwcub}. The overall algorithm runs in $\tilde{O}(n^{4/5}+D)$ rounds, dominated by the cost of the directed unweighted MWC subroutine.

\begin{proof}[Proof of Theorem \ref{thm:dirub}.\ref{thm:dirub:wt}]
    We prove that $(2+\epsilon)$-approximation of directed weighted MWC can be computed in $\tilde{O}(n^{4/5}+D)$ rounds.

    We make the following changes to Algorithm~\ref{alg:undirwtmwcub}. We modify line~\ref{alg:wtmwc:long} of Algorithm~\ref{alg:undirwtmwcub} to use the directed version of the approximate multiple source SSSP algorithm (Lemma~\ref{lem:n3ssspwt}) to compute long cycles. The round complexity of this line is unchanged, taking $\tilde{O}(n^{2/3}+D)$ rounds. We modify line~\ref{alg:wtmwc:scalemwc} to use a $h^*$-hop limited version of the directed unweighted MWC algorithm of Algorithm~\ref{alg:dirunwmwcub}. 
    
    We now modify Algorithm~\ref{alg:dirunwmwcub} to compute hop-limited MWC. We will restrict the BFS computations done by the unweighted algorithm to $h^*$ hops (in cases where BFS would have extended further), and the convergecast operations cost $O(D)$ where $D$ is the undirected diameter of the original weighted graph instead of the stretched scaled graph. 
    With these modifications, we can compute a 2-approximation of directed unweighted MWC of a stretched graph restricted to $h^*$-hops in $\tilde{O}(n^{4/5} + h^* + D)$ rounds. For more details, see the hop-limited modifications for the undirected version in the proof of Corollary~\ref{thm:hoplimapproxundirmwcub} in Appendix~\ref{app:undirunwmwcub}.

    Since scaling introduces an addition $(1+\epsilon)$-approximation factor, we compute a $(2+\epsilon)$-approximation of directed weighted MWC in $G$. The round complexity with these modifications is $\tilde{O}(n^{4/5}+D)$, dominated by the hop-limited directed unweighted MWC computation.
\end{proof}

\section{\texorpdfstring{Multiple Source SSSP from $k < n^{1/3}$ sources}{Multiple Source SSSP from k<n1/3 sources}}
\label{sec:ksssp}
We present our algorithm for computing directed BFS and SSSP from $k < n^{1/3}$ sources.
For a summary of our results, see Theorem~\ref{thm:ksssp}. Our algorithm smoothly interpolates between the current best round complexities for $k$-source approximate SSSP in directed graphs for $k=1$ (SSSP~\cite{cao2021approximatesssp}) and $k=n$ (APSP~\cite{bernsteinapsp}).  Our algorithm is sublinear whenever both $k$ and $D$ are sublinear. Our results may be of independent interest and have applications to other CONGEST algorithms.

For undirected weighted graphs, an algorithm in~\cite{elkin2019hopset} presents an $\tilde{O}(\sqrt{nk} \cdot n^{o(1)}+D)$ round CONGEST algorithm for $k$-source $(1+\epsilon)$-approximate shortest path for any $1 \le k \le n$. Their algorithm uses an undirected hopset construction on a sampled skeleton graph and uses broadcasts to communicate through hopset edges. In this section, we present a result for multiple source approximate SSSP in directed weighted graphs that utilizes a CONGEST algorithm for directed hopset construction given in \cite{cao2021approximatesssp} (which trivially applies to unweighted graphs as well). Our algorithm runs in $\tilde{O}(\sqrt{nk} + k^{2/5}n^{2/5+o(1)}D^{2/5} + D)$ rounds if $k<n^{1/3}$. For $k \ge n^{1/3}$, our previous algorithm in Section~\ref{sec:multiplesssp} is much more streamlined as it only involves broadcasts instead of the complicated directed hopset construction.

We now present the details of our algorithm, shown in Algorithm~\ref{alg:ksssp}. The algorithm computes approximate $k$-source shortest path distances in the skeleton graph, which gives us distances between sources and sampled vertices, using a directed hopset construction. For our skeleton graph shortest path algorithm, we first assume all edges in the skeleton graph are unweighted and directed, and show how to compute $k$-source directed BFS (Lemma~\ref{lem:skeleton}.A). Then, we use a scaling technique~(Fact~\ref{fact:approxhopsssp},\cite{nanongkai2014approx}) to extend the algorithm to $(1+\epsilon)$-approximate shortest path in the weighted skeleton graph (Lemma~\ref{lem:skeleton}.B). 

\begin{lemma}
    Let $G'=(S,E')$ be a directed weighted skeleton graph on directed $G=(V,E)$ ($S \subseteq V$ and $E' \subseteq S \times S$), where the underlying CONGEST network of $G$ has undirected diameter $D$ and $W$ is the maximum weight of an edge in $G$. Let $U \subseteq V$ be the set of sources, with $|U|=k$.
    \begin{enumerate}[label=\Alph*.]
        \item We can compute exact $k$-source $h$-hop directed BFS on the skeleton graph $G'$ in $O\left(k \cdot |S| + h\cdot D \right)$ rounds.
        \item Given an $(h,\frac{\epsilon}{2})$-hopset (for any constant $\epsilon > 0$) for the skeleton graph $G'$, we can compute $(1+\epsilon)$-approximate $k$-source weighted SSSP on the skeleton graph in $\tilde{O}\left((k|S| + hD) \cdot \log W \right)$ rounds.
    \end{enumerate}
    \label{lem:skeleton}
\end{lemma}
\begin{proof}
    \textbf{[A.]} Assuming the skeleton graph is undirected, we will compute distances $d(u,s)$ (will be known to node $s$) for each $u\in U,s \in S$. We assume that along with the skeleton graph, directed edges $(u,s)$ corresponding to $u$-$s$ directed paths in $G$ have been computed at node $s\in S$ if they exist. 
    
    Since we cannot communicate directly across edges of the skeleton graph (they correspond to hop-limited paths in the network), we will use broadcasts on the network $G$ to propagate BFS messages. Consider an $h$-hop BFS from one source $u \in U$. Each BFS message that an intermediate vertex $v$ in the skeleton graph sends will be of the form $(u, v, d)$, where $u$ is the source of the BFS and $d$ is the shortest $u$-$v$ distance. At the first step of the BFS, $u$ broadcasts the message $(u,u,0)$. All vertices in $S$ know their incoming edges in the skeleton graph $G'$, and the (unvisited) vertices $v$ with incoming edges $(u,v)$ will update their distance from $u$ to be 1. All such vertices will now broadcast the message $(u,v,1)$ and so on. If $N_{G'}(u,t)$ is the set of vertices at distance $t$ from $u$ in $G'$, all the vertices in $N_{G'}(u,t)$ will broadcast BFS messages at step $t$ of the BFS.

    Now, we consider all $k$ sources $u_1, \dots u_k \in U$. At a given step $t$, for each $i$  let $N_{G'}(u_i,t)$ be the set of vertices broadcasting BFS messages for source $u_i$. To communicate these BFS messages using a broadcast, assume that there is a broadcast tree that has been computed in the network $G$ with diameter $D$~\cite{peleg2000distributed}. The total number of messages to be broadcast in this step is $M_t = \sum_{i=1}^k |N_{G'}(u_i,t)|$, and we perform this broadcast by pipelining through the broadcast tree in $O(M_t + D)$ rounds. Now, we have at most $h$ steps in BFS (since we only compute BFS up to $h$ hops in $G'$), and for a single source $u_i$, $\sum_{t=1}^{h} |N_{G'}(u_i,t)|$ cannot exceed $|S|$, the number of vertices in $G'$. Using this, we bound $\sum_{t=1}^h (M_t + D)$ by $k|S| + hD$, and our round complexity for $k$-source $h$-hop BFS is $O(k |S| + hD)$ rounds.

    \textbf{[B.]} We use the scaling technique of~\cite{nanongkai2014approx} as stated in Fact~\ref{fact:approxhopsssp}, which allows us to compute bounded hop approximate shortest paths efficiently. We replace the $h$-hop BFS used in the unweighted algorithm with an $h$-hop $k$-source approximate SSSP algorithm on the skeleton graph, which uses $O(\log n)$ computations of $O(\frac{h}{\epsilon})$ hop BFS computations. This gives us a round complexity of $\tilde{O}\left((k|S| + hD) \cdot \log W \right)$ for computing $(1+\frac{\epsilon}{2})$-approximate SSSP in the skeleton graph augmented by the hopset. With the guarantee of the $(h,\frac{\epsilon}{2})$-hopset, the $h$-hop bounded approximate SSSP computations are  $(1+\epsilon)$-approximations of the shortest path distances in the skeleton graph.
\end{proof}

We now present our algorithm for computing $k$-source approximate SSSP in directed (unweighted or weighted) graphs and undirected weighted graphs. 

\begin{note}
    \label{note:ksssp}
    The proof of Theorem~\ref{thm:ksssp} establishes the following setting of parameters $|S|$ and $h$ for Algorithm~\ref{alg:ksssp} based on the input $k$ and $D$:

    \begin{minipage}{.45\linewidth}
        $$
    |S| = \begin{cases} 
        \sqrt{\frac{n}{k}} & ; k \ge n^{1/3} \text{ or } \\
        & \;\;\; k < n^{1/3}, D< n^{1/4}k^{3/4} \\
        \frac{n^{3/5}}{D^{2/5}} & ; k < n^{1/3}, n^{1/4}k^{3/4} < D < n^{2/3} \\
        n^{1/3} & ; k < n^{1/3},  n^{2/3} < D
    \end{cases}
    $$
    \end{minipage}%
    \hspace{0.06\linewidth}
    \begin{minipage}{.45\linewidth}
        $$
    h = \begin{cases} 
        1 & ; k \ge n^{1/3}  \\
        \frac{n^{1/4}}{k^{1/4}}& ; k < n^{1/3}, D< n^{1/4}k^{3/4} \\
        \frac{n^{2/5}}{D^{3/5}} & ; k < n^{1/3}, n^{1/4}k^{3/4} < D < n^{2/3} \\
        1 & ; k < n^{1/3},  n^{2/3} < D
    \end{cases}
    $$
    \end{minipage}
\end{note}

\begin{algorithm}[t]
    \caption{Approximate $k$-source Directed BFS algorithm}
    \begin{algorithmic}[1]
        \Require Directed unweighted graph $G=(V,E)$, set of sources $U \subseteq V$ with $|U| = k$.
        \Ensure Every vertex $v \in V$ computes $d(u,v)$ for each source $u \in U$.
        \State Let $h$, $|S|$ be parameters chosen based on $k$ and $D$ (undirected diameter of $G$) as in Note~\ref{note:ksssp}.
        \State Construct set $S \subseteq V$ by sampling each vertex $v \in V$ with probability $p = \Theta\left(\frac{|S| \cdot \log n}{n}\right)$ (parameter $p$ is chosen to achieve desired size $|S|$). \label{alg:ksssp:samp}
        \State  Perform directed BFS from each vertex in $S$ up to $\frac{n}{|S|}$ hops, which takes $O(|S| + \frac{n}{|S|})$ rounds. Repeat this computation in the reversed graph. \label{alg:ksssp:sampbfs}
        \State Perform $\frac{n}{|S|}$-hop directed BFS from each of the $k$ sources, which takes $O(k+\frac{n}{|S|})$ rounds. \\ \rightComment{Compute short hop distances.} \label{alg:ksssp:kbfs}
        \State Construct a skeleton graph on the vertices in $S$, with an edge for each directed shortest path with at most $\frac{n}{|S|}$ hops found during the BFS in line~\ref{alg:ksssp:sampbfs}. \rightComment{Local computation.} \label{alg:ksssp:skeleton}
        \State Construct an $(h,\epsilon)$-directed hopset on the skeleton graph (parameter $h$ to be set as in Note~\ref{note:ksssp}). This takes $\tilde{O}(\frac{|S|^2}{h^2} + h^{1+o(1)}D)$ rounds, for $h \le \sqrt{|S|}$ using the hopset algorithm of~\cite{cao2021approximatesssp}. \label{alg:ksssp:hopset}
        \State Compute a $k$-source $h$-hop approximate SSSP on the skeleton graph, with the set $U$ as sources, treating $h$-hop paths from $U$ to $S$ (computed in line~\ref{alg:ksssp:kbfs}) as a single edge: each sampled vertex $s$ now knows $(1+\epsilon)$-approximate distance $d(u,s)$ for each $u \in U$. \\ \rightComment{Uses Lemma~\ref{lem:skeleton}.B.} \label{alg:ksssp:skeletonapsp} 
        \State Each sampled vertex $s$ propagates the distance $d(u,s)$ for each $u \in U$ through the BFS tree rooted at $s$ computed in line~\ref{alg:ksssp:sampbfs}. Using randomized scheduling~\cite{ghaffarischeduling}, this takes $\tilde{O}(\frac{n}{|S|}+k|S|)$ rounds. All vertices $v$ reached by the BFS now compute their $(1+\epsilon)$-approximate distance from source $u$ as $d(u,v) \gets \min(d(u,v), d(u,s)+d(s,v))$.  \label{alg:ksssp:propbfs}
    \end{algorithmic}
    \label{alg:ksssp}
\end{algorithm}

\begin{proof}[Proof of Theorem~\ref{thm:ksssp}]
    We present our algorithm for approximate shortest paths in directed unweighted graphs in Algorithm~\ref{alg:ksssp}. The correctness follows from arguments similar to Lemma~\ref{lem:n3bfs}. The main difference is how we compute SSSP in the skeleton graph -- we use a hopset construction in this algorithm while we broadcast all skeleton graph edges in the $n^{1/3}$-source SSSP algorithm. We construct a $(h,\epsilon)$-directed hopset in line~\ref{alg:ksssp:hopset}, so in line~\ref{alg:ksssp:skeletonapsp} the $h$-hop BFS computes $(1+\epsilon)$-approximate shortest paths in the skeleton graph. As in Lemma~\ref{lem:n3bfs}, propagating distances in line~\ref{alg:ksssp:propbfs} is sufficient to compute approximate shortest path distances to all vertices by our sampling probability.
    
    We now analyze the number of rounds. The algorithm samples a vertex set $S$ (whose size is a parameter to be fixed later) and performs BFS computations restricted to $\frac{n}{|S|}$ hops. On the skeleton graph built on these sampled vertices, we construct a $(1+\epsilon)$-approximate hopset with hop bound $h$ (this parameter is fixed later). We use the approximate hopset algorithm of~\cite{cao2021approximatesssp}, which constructs a $\left(\frac{|S|^{1/2+o(1)}}{\rho}, \epsilon\right)$-hopset in $O(|S|\frac{\rho^2}{\epsilon^2}\log W + \frac{|S|^{1/2+o(1)}D}{\rho \epsilon}\log W)$ rounds. We denote the hop bound by $h=\frac{|S|^{1/2}}{\rho}$, and get a round complexity of $O\left(\left( |S| \cdot \frac{|S|}{h^2} +h^{1+o(1)}D \right) \cdot \frac{\log W}{\epsilon^2}\right)$. Computing $k$-source approximate SSSP after this requires $\tilde{O}(k|S|+h^{1+o(1)}D)$ rounds using Lemma~\ref{lem:skeleton}. The total round complexity of the algorithm is $\tilde{O}(\frac{n}{|S|} + \frac{|S|^2}{h^2} + h^{1+o(1)}D +k|S|)$. We set the parameters $|S|$ and $h$ based on the values of $D$ and $k$ as shown in Note~\ref{note:ksssp}.

    Note that for any setting of the parameter $|S|$, the quantity $\frac{n}{|S|}+k|S|$ is at least $\sqrt{nk}$. For $k \ge n^{1/3}$, the parameter setting $|S|=\sqrt{\frac{n}{k}} , h=1$ achieves round complexity $\tilde{O}(\sqrt{nk} + D + \frac{n}{k})$ which is $\tilde{O}(\sqrt{nk} + D)$ since $\frac{n}{k} \le \sqrt{nk}$ for $k \ge n^{1/3}$. We also note that $D$ is a lower bound for any $k$-source SSSP algorithm. Hence, this round complexity is the best we can achieve for any parameter choice for $k\ge n^{1/3}$. Note that $h=1$ essentially means all pairs shortest paths are computed at all vertices, recovering the algorithm in Section~\ref{sec:multiplesssp}.

    For $k<n^{1/3}$, we need to set parameters based on $D$ to minimize the round complexity $\tilde{O}(\frac{n}{|S|} + \frac{|S|^2}{h^2} + h^{1+o(1)}D +k|S|)$, with the constraints $1 \le |S| \le n, 1 \le h \le \sqrt{|S|}$.
    \begin{enumerate}
        \item When $D$ is small ($D$ is $o(n^{1/4}k^{3/4})$), the term $\frac{|S|^2}{h^2} + h^{1+o(1)}D$ is minimized by setting $h$ to its maximum value $\sqrt{S}$. With this choice of $h$, the round complexity is $\tilde{O}(\frac{n}{|S|} + |S| + \sqrt{|S|}D +k|S|)$, which is minimized by setting $|S|=\sqrt{\frac{n}{k}}$ assuming $\sqrt{|S|}D$ is small compared to the other terms. For $D < n^{1/4}k^{3/4}$, we have $ \frac{n^{1/4}}{k^{1/4}} D < \sqrt{nk}$, and the total round complexity is $\tilde{O}(\sqrt{nk})$.
        
        \item When $D$ is large($D$ is $\Omega(n^{2/3})$), we minimize the term $h^{1+o(1)}D$ by setting $h=1$ and the resulting round complexity expression is $\tilde{O}(\frac{n}{|S|} + |S|^2 + D +k|S|)$. Since we are only concerned with the case $k< n^{1/3}$, we minimize the expression $\frac{n}{|S|} + |S|^2$ by setting $|S|=n^{1/3}$ giving round complexity $\tilde{O}(n^{2/3}+D+kn^{1/3})$. For the case $k<n^{1/3}, D>n^{2/3}$, this is $\tilde{O}(D)$.
        
        \item For intermediate $D$, i.e , $n^{1/4}k^{3/4} < D < n^{2/3}$, we balance terms $\frac{|S|^2}{h^2}$, $\frac{n}{|S|}$ and $h^{1+o(1)}D$, giving the parameters $|S| = \frac{n^{3/5}}{D^{2/5}}, h = \frac{n^{2/5}}{D^{3/5}}$, and the round complexity expression is $\tilde{O}(n^{2/5+o(1)}D^{2/5} + k\frac{n^{3/5}}{D^{2/5}})$. For $D>n^{1/4}k^{3/4}$, we have $k\frac{n^{3/5}}{D^{2/5}} \le kn^{3/5}D^{2/5} \cdot \frac{1}{D^{4/5}} \le kn^{3/5}D^{2/5} \frac{1}{n^{1/5}k^{3/5}} = k^{2/5}n^{2/5}D^{2/5}$. So, we can rewrite our round complexity as $\tilde{O}(n^{2/5+o(1)}D^{2/5} + k^{2/5}n^{2/5}D^{2/5})$
    \end{enumerate}

    Our final round complexity is $\tilde{O}(\sqrt{nk} + D)$ if $k \ge n^{1/3}$ and $\tilde{O}(\sqrt{nk} + k^{2/5}n^{2/5+o(1)}D^{2/5} + D)$ if $k<n^{1/3}$. This round complexity is sublinear whenever $k$ and $D$ are sublinear, and beats the $\tilde{O}(n)$ round APSP algorithm for $k=o(n)$. For $k=1$, we recover the current best SSSP algorithm taking $\tilde{O}(\sqrt{n}+n^{2/5+o(1)}D^{2/5} + D)$ rounds~\cite{cao2023sssp}, hence our algorithm smoothly interpolates between SSSP and APSP algorithms for $k=1$ and $k=n$ respectively.

    \paragraph*{\textbf{Weighted Graphs}}
    We extend the multiple source directed BFS algorithm described in Algorithm~\ref{alg:n3sssp} to approximate multiple source SSSP by using scaling. We use the algorithm of~\cite{nanongkai2014approx} stated in Fact~\ref{fact:approxhopsssp} to compute $(1+\epsilon)$-approximate $h$-hop SSSP from $k$ sources in $\tilde{O}(h+k+D)$ rounds, and the approximate SSSP algorithm for the skeleton graph described in Lemma~\ref{lem:skeleton}.B. Thus, the BFS in lines~\ref{alg:ksssp:sampbfs},\ref{alg:ksssp:kbfs} takes $\tilde{O}(|S|+\frac{n}{|S|})$ and $\tilde{O}(k+\frac{n}{|S|})$ rounds respectively. The distances $d(s,v)$,$d(u,s)$ computed in lines~\ref{alg:ksssp:kbfs},\ref{alg:ksssp:skeletonapsp} are still $(1+\epsilon)$-approximations of the shortest path distance, and hence the final $d(u,v)$ distances computed in line~\ref{alg:ksssp:propbfs} are also $(1+\epsilon)$-approximations. Thus. we compute $k$-source weighted $(1+\epsilon)$-approximate SSSP with the same round complexity as in unweighted graphs up to $\log$ factors.
\end{proof}

\section{Conclusion and Open Problems}
\label{sec:conclusion}

We have presented several CONGEST upper and lower bounds for computing approximate MWC in directed and undirected graphs,  both weighted and unweighted. Here are some topics for further research. 

For $(2-\epsilon)$-approximation of girth, a $\tilde{\Omega}(\sqrt{n})$ lower bound is known~\cite{frischknecht2012} which we nearly match with a $\tilde{O}(\sqrt{n}+D)$ for $(2-(1/g))$-approximation. For larger approximations, we present lower bounds of $n^{1/3}$ for $(2.5-\epsilon)$-approximation and $n^{1/4}$ for arbitrarily large constant $\alpha$-approximation. The current best upper bound for girth is still $O(n)$~\cite{holzer2012apsp, censorhillel2020girth}, and the CONGEST complexity of exact girth remains an open problem.

We have studied larger constant $(\alpha \ge 2)$ approximation of MWC in directed graphs and weighted graphs, and presented sublinear algorithms for 2-approximation ($(2+\epsilon)$ for weighted graphs) beating the linear lower bounds for $(2-\epsilon)$-approximation. Our results include: for directed unweighted MWC, $\tilde{O}(n^{4/5}+D)$-round 2-approximation algorithm; for directed weighted MWC,  $(2+\epsilon)$-approximation algorithm with the same $\tilde{O}(n^{4/5}+D)$ complexity; for undirected weighted MWC, $\tilde{O}(n^{2/3}+D)$-round algorithm for $(2+\epsilon)$-approximation. For these three graph types, we showed lower bounds of $\tilde{\Omega}(\sqrt{n})$ for any $\alpha$-approximation of MWC, for arbitrarily large constant $\alpha$. Whether we can bridge these gaps between upper and lower bounds, or provide tradeoffs between round complexity and approximation quality is a topic for further research.
    
Our approximation algorithms for weighted MWC (directed and undirected) are based on scaling techniques, which introduce an additional multiplicative error causing our algorithms to give $(2+\epsilon)$-approximation instead of the 2-approximation obtained in the unweighted case. The main roadblock in obtaining a 2-approximation is an efficient method to compute exact SSSP from multiple sources, on which we elaborate below. 

When $k\ge n^{1/3}$, we have presented a fast and streamlined $\tilde{O}(\sqrt{nk}+D)$-round algorithm for $k$-source exact directed BFS, where the key to our speedup is sharing shortest path computations from different sources using skeleton graph constructions. Using scaling techniques, we extended this to an algorithm for $k$-source directed SSSP in weighted graphs, but only for $(1+\epsilon)$-approximation.
While there have been recent techniques for a single source to compute exact SSSP from approximate SSSP algorithms~\cite{cao2023sssp,rozhon22sssp} building on~\cite{klein1997randomized}, it is not clear how to extend them to multiple sources seems difficult. 
These techniques involve distance computations on graphs whose edge-weights depend on the source. As a result, we can no longer construct a single weighted graph where we can share shortest computations for $k$ sources. Providing an exact $k$-source SSSP algorithm that matches the round complexity of our $k$-source approximate weighted SSSP algorithm is a topic for further research.

For undirected weighted graphs, a $k$-source approximate SSSP that runs in $\tilde{O}((\sqrt{nk}+D)\cdot n^{o(1)})$ rounds (and $\tilde{O}(\sqrt{nk}+D)$ for $k=n^{\Omega(1)}$) is given in~\cite{elkin2019hopset}. For directed weighted graphs, we match this bound with a simpler algorithm when $k \ge n^{1/3}$, that works for both directed and undirected graphs. When $k < n^{1/3}$, our algorithm takes $\tilde{O}(\sqrt{nk} + k^{2/5}n^{2/5+o(1)}D^{2/5} + D)$. This difference is mainly because the directed hopset construction of~\cite{cao2021approximatesssp} has higher round complexity than the undirected case~\cite{elkin2019hopset}, and it is an open problem whether this gap can be closed.

\appendix

\section{Approximate Undirected Unweighted MWC}
\label{app:undirmwc}
We present detailed pseudocode and proofs for the algorithm described in Section~\ref{sec:undirunwtmwcub} for $(2-\frac{1}{g})$-approximation of undirected unweighted MWC (girth).

\label{app:undirunwmwcub}
\begin{algorithm}[t]
    \caption{$(2-\frac{1}{g})$-Approximation Algorithm for Undirected Unweighted MWC (Girth)}
    \begin{algorithmic}[1]
        \Require Undirected unweighted graph $G(V,E)$
        \Ensure $M$, $(2-\frac{1}{g})$-approximation of MWC of $G$ where $g$ is the length of MWC.
        \State $\forall v \in V, M_v \gets \infty$. \rightComment{$M_v$ denotes the best cycle found so far at vertex $v$.}
        \State Construct set $S \subseteq V$ by sampling each vertex $v \in V$ with probability $\Theta(\frac{\log n}{ \sqrt{n}})$. \label{alg:undirunwmwc:sample}
        \ForAll{vertex $w \in S$}
            \lineComment{Compute MWC through sampled vertices}
            \State Perform BFS starting from $w$: Every vertex $z \in V$ receives a message from $w$ containing distance value $d(w,z)$. \label{alg:undirunwmwc:bfs}
            \State For each $z \in V$ for which $z$ receives two messages for the BFS rooted at $w$ from two distinct neighbors $x$,$y$ with distance values $d_1 = d(w,x)+1$ and $d_2 = d(w,y)+1$, vertex $z$ locally records a cycle of weight $d_1 + d_2$, i.e. $M_z \gets \min (M_z, d_1+d_2)$. \label{alg:undirunwmwc:bfscyc}
        \EndFor
        \ForAll{vertex $v \in V$}
            \lineComment{Compute approximate MWC within neighborhood of $v$}
            \State Compute shortest path distance from $v$ to its $\sqrt{n}$ closest vertices ,i.e. the $\sqrt{n}$-neighborhood of $v$ denoted $Q(v)$, using a source-detection algorithm~\cite{lenzen2019distributed}: $v$ receives a message from each $z \in Q(v)$ containing distance value $d(v,z)$. Vertex $v$ also remembers its neighbor $u$ that forwarded this message from each such $z \in Q(v)$, and we denote $p(v,z) = u$. \label{alg:undirunwmwc:neighb}
            \State For each $z \in Q(v)$ for which $v$ receives two messages from two distinct neighbors $x$,$y$ with distance values $d_1 = d(x,z)+1$ and $d_2 = d(y,z)+1$, vertex $v$ records a cycle of weight $d_1+d_2$, i.e, $M_v \gets \min(M_v, d_1+d_2)$.  \label{alg:undirunwmwc:neighbcyc} \rightComment{Compute cycles through $v$ completely contained in $Q(v)$}
            \State Send $\sqrt{n}$-neighborhood information $\{(z, d(v,z), p(v,z)) \mid z \in Q(v)\}$ to neighbors $N(v)$. \label{alg:undirunwmwc:neighbsend} 
            \lineComment{Lines \ref{alg:undirunwmwc:neighbcyc2},\ref{alg:undirunwmwc:neighbcyc3} locally compute cycles through $v$ with exactly one vertex outside $Q(v)$, using distances sent to $v$ in line~\ref{alg:undirunwmwc:neighbsend}.}
            \State Receives $\sqrt{n}$-neighborhood information from each neighbor $u$, and locally checks if there is a vertex $z$ such that $z \in Q(v)$ and $z \in Q(u)$ and $p(v,z) \ne u, p(u,z) \ne v$. If so. $v$ records a cycle: $M_v \gets \min(M_v, d(v,z) + d(u,z) + 1)$. \label{alg:undirunwmwc:neighbcyc2}
            \State Locally checks if there are a pair of neighbors $x,y \in N(v)$ and a vertex $z$ such that $z \in Q(x)$, $z \in Q(y)$ and $p(x,z) \ne p(y,z), p(x,z) \ne y, p(y,z) \ne x$. If so, $v$ records a cycle: $M_v \gets \min(M_v, d(x,z) + d(y,z) + 2)$. \label{alg:undirunwmwc:neighbcyc3} 
        \EndFor
        \State Compute $M = \min_{z \in V} M_z$ by a convergecast operation. \label{alg:undirunwmwc:globalmin}
    \end{algorithmic}
    \label{alg:undirunwmwcub}
\end{algorithm}

\begin{lemma}
    Algorithm~\ref{alg:undirunwmwcub} computes a $(2-\frac{1}{g})$-approximation of MWC of $G$, where $g$ is the weight of the MWC, in $\tilde{O}(\sqrt{n}+D)$ rounds.
    \label{lem:undirunwmwcub}
\end{lemma}
\begin{proof}
    \noindent
    \textbf{{Correctness:}}
    Let $C$ be a minimum weight cycle of $G$ with weight $w(C)=g$. For vertex $v$, $Q(v)$ denotes the $\sqrt{n}$-neighborhood of $v$, i.e., set of $\sqrt{n}$ closest vertices with ties broken by a fixed global ordering. We consider the cases where $C$ has even or odd weight separately.
    
    \vspace{0.2cm}
    \noindent {\boldmath\textbf{\textit{Case 1: $C$ has even weight}}}, $w(C) = g = 2r$. We have three cases:
    \vspace{0.2cm}
    
    {\boldmath\textbf{\textit{Case 1.a: For every $v \in C$, $C$ is contained entirely within $Q(v)$ :}}} Choose a vertex $v \in C$, and consider the vertex $z$ in $C$ that is furthest from $v$, we know $z \in Q(v)$ and $d(v,z)=r$. Let vertices $x,y$ denote the neighbors of $v$ in $C$. Given that $x$ is on a shortest path from $z$ to $v$, $z$ must be in the $\sqrt{n}$-neighborhood $Q(x)$ of $x$. Otherwise, we have $\sqrt{n}$ vertices other than $z$ at distance $\le d(v,z)-1$ from $x$ which means these $\sqrt{n}$ vertices are at distance $\le d(v,z)$ from $v$ contradicting $z \in Q(v)$ (these other vertices are before $z$ in global ordering in case of tied distance). Thus, we have $z \in Q(x)$ and $z \in Q(y)$ and in line~\ref{alg:undirunwmwc:neighbcyc}, $v$ would receive messages from $x$ and $y$ for source $z$ with distance value $r$, and $v$ records $M_v = 2r = g$ .
    
    {\boldmath\textbf{\textit{Case 1.b: For every $v \in C$, $C$ has at most one vertex outside $Q(v)$:}}} Consider a vertex $v \in C$ with vertex $z \in C$ outside $Q(v)$. Note that this $z$ must necessarily be the furthest vertex from $v$, at a distance $r-1$ from it. Consider the neighbors $x$ and $y$ of $v$ in $C$. We must have $z \in Q(x)$ as otherwise both $z$ and the furthest vertex from $x$ (that is distinct from $z$ as $C$ is an even cycle) are outside $Q(x)$. Similarly $z \in Q(y)$, satisfying the condition in line~\ref{alg:undirunwmwc:neighbcyc3} where $v$ receives distances $d(x,z)=(r-1)$ and $d(y,z)=(r-1)$ and records $M_v = 2r = g$.
    
    {\boldmath\textbf{\textit{Case 1.c: For some $v \in C$, $C$ has more than one vertex outside $Q(v)$:}}} In this case, $Q(v)$ contains $\sqrt{n}$ vertices, as otherwise it would contain all vertices connected to $v$. Since we sample $S$ with probability $\Theta(\frac{\log n}{ \sqrt{n}})$, w.h.p. in $n$, $Q(v)$ contains a sampled vertex $w \in S$. By our assumption, let $z,z'$ be two vertices on $C$ such that $z,z' \not\in Q(v)$ with $d(v,z) < d(v,z') \le r$ (WLOG). Since $w$ is one of the $\sqrt{n}$ closest vertices to $v$ but $z$ is not, we have $d(w,v) \le d(v,z) \le r-1$. 
    
    Now, consider the BFS tree rooted at $w$ computed in line~\ref{alg:undirunwmwc:bfs}. The cycle $C$ contains some non-tree edge $(x',y')$ of this BFS tree. Assuming WLOG $d(x',w) \ge d(y',w)$, $x'$ receives two messages from $w$ from distinct neighbors with distance values $d(w,x')$ and $d(w,y')+1$, and $x'$ records a cycle $M_{x'} \le  d(w,x')+d(w,y')+1$ in line~\ref{alg:undirunwmwc:bfscyc}. We have $d(w,x') \le d(w,v) + d(v,x') \le r+(v,x')$ and similarly $d(w,y') \le r+(v,x')$. Since $v,x',y'$ are all on the cycle $C$ of length $2r$, we have $d(v,x') + d(v,y') + 1 \le 2r$. Substituting, we get $M_{x'} \le d(v,x') + d(v,y') + 1 + 2(r-1) \le 4r-2 \le 2g-2$. So, a cycle of length at most $2g-2$ is recorded at $x'$.

    \vspace{0.2cm}
    \noindent {\boldmath\textbf{\textit{Case 2. $C$ has odd weight}}}, $w(C) = g = 2r+1$. We have two cases:
    \vspace{0.2cm}
    
    {\boldmath\textbf{\textit{Case 2.a: For every $v \in C$, $C$ is contained entirely within $Q(v)$:}}} Let vertex $z$ be one of the vertices in $C$ furthest from $v$, we have $z \in Q(v)$ and $d(v,z)=r$. Let $u$ be the neighbor of $v$ in $C$ that is further from $z$, and we have $d(u,z)=r$. With our assumption $z \in Q(u)$, and we satisfy the condition of line~\ref{alg:undirunwmwc:neighbcyc2}, and $v$ records a cycle $M_v = d(v,z) +d(u,z) + 1 = 2r+1 = g$. 
    
    {\boldmath\textbf{\textit{Case 2.b: For some $v \in C$, $C$ has a vertex outside $Q(v)$:}}} With the same reasoning as case \textbf{1.c}, $Q(v)$ contains at least $\sqrt{n}$ vertices and hence $Q(v)$ contains a sampled vertex $w \in S$. Since we have a vertex $z \not\in Q(v)$, we have $d(w,v) \le d(z,v) \le r$. Consider the edge $(x',y')$ in $C$ that is a non-tree edge of the BFS tree rooted at $w$ computed in line~\ref{alg:undirunwmwc:bfs}, and $x'$ records a cycle of length $M_x' \le d(w,x')+d(w,y')+1 \le 2r+1 + 2d(w,v) \le 4r + 1 = 2g-1$ in line~\ref{alg:undirunwmwc:bfscyc}.

    \vspace{0.3cm}
    We also prove that any update we make to $M_z$ corresponds to at least the weight of some simple cycle in $G$. In line~\ref{alg:undirunwmwc:neighbcyc}, we have a cycle made of two shortest paths $x$-$z$ and $y$-$z$ along with edges $(x,v)$ and $(y,v)$ with $v \ne z$ and $x \ne y$, which means there is a simple cycle of weight at most $d(x,z) + d(y,z) + 2$. In line~\ref{alg:undirunwmwc:neighbcyc2}, we have a cycle made of two shortest paths $v$-$z$ and $u$-$z$ along with edge $(v,u)$ with $v \ne u$ and neither $v$ nor $u$ are on the shortest path from $z$ to $u$ or $v$ respectively. So, there is a simple cycle of weight at most  $d(v,z) + d(u,z) + 1$. A similar argument holds for cycles computed in line~\ref{alg:undirunwmwc:neighbcyc3}. In line~\ref{alg:undirunwmwc:bfscyc}, we have two shortest paths $x$-$w$ and $y$-$w$ along with edges $(x,z),(y,z)$ with $x \ne y$. Note that we allow $x=w$ or $y=w$, and we still have a cycle of weight at most $d(w,x)+d(w,y)+2$. Thus, the returned value $M$ is at least the weight of some simple cycle in $G$. Combining the even and odd cycle cases, we compute a cycle of weight at most $2g-1$, which is a $(2-\frac{1}{g})$-approximation of MWC.

    \noindent
    \textbf{{Running Time:}}
    Now, we show that Algorithm~\ref{alg:undirunwmwcub} runs in  $\tilde{O}(\sqrt{n}+D)$ rounds. The set $S$ computed in line~\ref{alg:undirunwmwc:sample} has size $\tilde{O}(\sqrt{n})$ w.h.p. in $n$  and the computation in line~\ref{alg:undirunwmwc:bfs} is an $|S|$-source BFS of at most $D$ hops, which takes $O(|S|+D) = \tilde{O}(\sqrt{n}+D)$ rounds~\cite{lenzen2019distributed}. Line~\ref{alg:undirunwmwc:bfscyc} only involves local computations.
    
    We implement line~\ref{alg:undirunwmwc:neighb} using the source detection algorithm in~\cite{lenzen2019distributed}, which computes the $R$ closest sources within $h$-hops for each vertex in $O(R+h)$ rounds. We choose our set of sources to be the entire graph, and choose $R=\sqrt{n}$, $h=D$. We also note that the message containing distance $d(v,z)$ is passed to $v$ by the neighbor of $v$ on a $v$-$z$ shortest path, so $v$ can keep track of $p(v,z)$. Thus, we implement line~\ref{alg:undirunwmwc:neighb} in $O(\sqrt{n}+D)$ rounds. Line~\ref{alg:undirunwmwc:neighbcyc} is just a local computation, and line~\ref{alg:undirunwmwc:neighbsend} sends $O(\sqrt{n})$ information along each edge taking $O(\sqrt{n})$ rounds. Lines~\ref{alg:undirunwmwc:neighbcyc2},\ref{alg:undirunwmwc:neighbcyc3} are local computations after information sent by neighbors in line~\ref{alg:undirunwmwc:neighbsend} has been received.

    Finally, for the global minimum computation in line~\ref{alg:undirunwmwc:globalmin}, we use a convergecast operation taking $O(D)$ rounds~\cite{peleg2000distributed}. Thus, the total round complexity of our algorithm is $\tilde{O}(\sqrt{n}+D)$.
\end{proof}

We now prove Corollary~\ref{thm:hoplimapproxundirmwcub} (repeated below) for computing $h$-hop limited MWC in stretched unweighted graphs. Recall that the stretched unweighted graph $G^s$ of a weighted graph $G$ is obtained by replacing each weight-$w$ edge in $G$ with an unweighted path of $w$ edges.

\begin{corollary}
    Given a network $G=(V,E)$ with edge weights, we can compute a $(2-1/g)$-approximation of $h$-hop limited MWC of $G^s$ (stretched unweighted graph of $G$) in $\tilde{O}(\sqrt{n}+h+R_{cast})$ rounds, where $g$ is the $h$-hop limited MWC value and $R_{cast}$ is the round complexity of a convergecast operation in $G$.
\end{corollary}
\begin{proof}
    We restrict the computation performed in Algorithm~\ref{alg:undirunwmwcub} to $h$ hops. Specifically, the BFS in line~\ref{alg:undirunwmwc:bfs} is restricted to $h$ hops, which takes $\tilde{O}(\sqrt{n}+h)$ rounds since we have $\tilde{O}(\sqrt{n})$ sources. We restrict the source detection computation in line~\ref{alg:undirunwmwc:neighb} to $h$ hops, and the algorithm of~\cite{lenzen2019distributed} computes the required distances in $O(\sqrt{n}+h)$ rounds. We communicate $O(\sqrt{n})$ amount of neighborhood information in line~\ref{alg:undirunwmwc:neighbsend}, and line~\ref{alg:undirunwmwc:globalmin} takes $R_{cast}$ rounds to perform the global minimum. If $C$ is the minimum weight cycle of at most $h$ hops, we can verify that for each case in the proof of Lemma~\ref{lem:undirunwmwcub}, $C$ is correctly detected when the operations are replaced by the $h$-hop limited versions. Thus, we compute a $(2-1/g)$-approximation of $C$ in $\tilde{O}(\sqrt{n}+h+R_{cast})$ rounds.
\end{proof}

\bibliographystyle{plainurl}
\bibliography{main-mwc} 

\end{document}